\documentclass[12pt,twoside,letter]{amsart}
\usepackage{txfonts}
\usepackage{bbm}
\usepackage{mathrsfs}
\usepackage{amsfonts}
\usepackage{amsmath}
\usepackage{amssymb}
\usepackage{wasysym}
\usepackage{psfrag}
\usepackage{graphics,epsfig,amsmath}  
\usepackage{comment}
\usepackage{enumerate}

\newtheorem{theorem}{Theorem}
\newtheorem{definition}[theorem]{Definition}
\newtheorem{algorithm}[theorem]{Algorithm}
\newtheorem{lemma}[theorem]{Lemma}
\newtheorem{preremark}[theorem]{Remark}
\newenvironment{remark}{\begin{preremark}\rm}{\end{preremark}}

\def\E{\mathcal{E}}
\def\A{\mathscr{A}}
\def\ep{\epsilon}
\def\U{\widehat{U}}
\def\h{\hat{h}}

\title[KAM theory for 1-D  quasicrystals]
{KAM theory for quasi-periodic equilibria in 1-D quasiperiodic media}
\author[X. Su] {Xifeng Su}
\address{Dept. of Mathematics\\
 Nanjing University\\
 Nanjing, 210093, CHINA}
\address{Dept. of Mathematics\\
Univ. of Texas at Austin \\
1 Univ. Station C1200\\
Austin, TX 78712-0257 }
\email{billy3492@gmail.com,
xifengsu@math.utexas.edu}
\author[R. de la Llave]{Rafael de la Llave}
\address{Dept. of Mathematics\\
Univ. of Texas at Austin \\
1 Univ. Station C1200\\
Austin, TX 78712-0257 }
\email{llave@math.utexas.edu}

\begin{document}
\maketitle
\begin{abstract}
We consider Frenkel-Kontorova models corresponding to 1 dimensional
quasicrystals.

We present a KAM theory for quasi-periodic equilibria. The theorem
presented has an \emph{a-posteriori} format. We show that, given an
approximate solution of the equilibrium equation, which satisfies
some appropriate non-degeneracy conditions, then, there is a true
solution nearby. This solution is locally unique.

Such a-posteriori theorems can be used to validate numerical
computations and also lead immediately to several consequences a)
Existence to all orders of perturbative expansion and their
convergence b) Bootstrap for regularity c) An efficient method to
compute the breakdown of analyticity.

Since the system does not admit an easy dynamical formulation, the
method of proof is based on developing several identities. These
identities also lead to very efficient  algorithms.
\end{abstract}



\section{Introduction}
We consider Frenkel-Kontorova models with a quasi-periodic
potential. In these models, the state of a system is given by a
sequence $\{q_i\}_{i\in \mathbb{Z}}$ of numbers and the physical
states are selected to be the critical points of a formal energy
\begin{equation}\label{FK energy}
\mathscr{S}(\{u\}_{i\in
\mathbb{Z}})=\sum_{n\in\mathbb{Z}}\frac{1}{2}(u_n-u_{n+1})^2-V(u_n)
\end{equation}

The critical points of $\mathscr{S}$ are obtained by taking formal
derivatives of $\mathscr{S}$ and setting them to zero
$\frac{\partial}{\partial
u_n}\mathscr{S}(\{u\}_{i\in\mathbb{Z}})=0$. That is
\begin{equation}\label{equilibrium}
u_{n+1}+u_{n-1}-2u_n+V'(u_n)=0.
\end{equation}

In the classical Frenkel-Kontorova model
\cite{FK, ALD}, $V$ is a periodic function
$V(u_n+1)=V(u_n)$--we choose units so that the period is normalized
to 1.

In this paper, however, we choose $V$ to be a quasi-periodic
function \[V(\theta)=\widehat{V}(\theta\alpha)\] where $\hat{V}:
\mathbb{T}^d\rightarrow\mathbb{R}$ and
$\theta\in\mathbb{R},~\alpha\in \mathbb{R}^d$ will be an irrational
vector, that is $k\cdot \alpha\neq 0$ when $k\in\mathbb{Z}^d-\{0\}$
where $d\geq 2$. Later, we will also assume further non-resonance
conditions. One example to keep in mind could be
\[
V(\theta)=a\cdot \sin(2\pi \theta)+b\cdot \sin(2\pi \sqrt{2}\theta)
\]where $a,b\in \mathbb{R}$.

Our goal in this paper is to prove an analogue of KAM theorem.
See Theorem~\ref{main}, \ref{main2}. Under some appropriate conditions,
we show that there are smooth families of quasi-periodic
solutions. See Section~\ref{sec:hull} for precise definitions of
these families.  The proof of the theorems is rather constructive
and leads to efficient algorithms that are being implemented.

There are several physical interpretations of FK models. The
original motivation \cite{FK} was dislocations in solids. In the
interpretation of \cite{ALD} $u_i$ are the positions of a deposited
material over a substratum. The interaction of the atoms with the
substratum is modeled by the term $V$. The periodicity of $V$
considered in \cite{ALD} corresponds to a periodic substratum (e.g.
a crystal) and the quasi-periodic models considered here could
appear in cleaved faces of crystals or in quasi-crystals.

In the interpretation of deposition, the existence of quasiperiodic
solutions
implies the existence of a continuum of equilibria, so that the
system can \emph{slide}. In contrast, if the KAM tori are not
present, the system is pinned. There have been numerical
explorations of these issues in \cite{vanErpthesis, VanErp'99a,
VanErp'01, vanErp'02, Radulescu-Jansen, Radulescu-Jansen2}. In
particular, the above references pay special attention to the
boundary of the set of parameters for which there is
an analytic solution (breakdown of analyticity,
Aubry transition), which corresponds physically
to the boundary between sliding and pinning.

Note that the physical argument to obtain sliding only requires a
continuous family of solutions, whereas it has been found that often
the boundary is described by the breakdown of an analytic family. In
Section~\ref{bootstrap of regularity} we will show that all the
families which are smooth enough are indeed analytic. {From} the
mathematical point of view this leaves open the possibility that
there are regimes where the solutions have only a rather low
regularity. In the case of twist mappings, using renormalization
group, this regime has found to exist but being a codimension one
surface \cite{Koch-boundary}.

In the mathematical literature, quasi-periodic Frenkel-Kontorova
models have been considered in \cite{Gambaudo, Aliste-Prieto}, which
use mainly topological methods to study the existence of orbits with
rotation number. In the periodic case, the critical points of the
energy, i.e. the configurations solving \eqref{equilibrium} can be
identified with orbits of a dynamical system on the annulus
$\mathbb{T}\times\mathbb{R}$. In the quasi-periodic case, no such
identification is easy. \cite{Gambaudo} shows that the
quasi-periodic case can be considered as a dynamical system on a
Cantor set ( Delone set).

To the best of our knowledge, there is no
systematic Aubry-Mather theory analogue to that of the periodic 1-D
Frenkel-Kontorova systems (existence of minimizing Aubry-Mather well
ordered minimizers). It seems possible that one could get analogues
of the theory of minimizing invariant measures.
See  \cite{Burkov87, Burkov88, Burkov90}.

\subsection{Hull function}
\label{sec:hull}

We will be interested in solutions of
\eqref{equilibrium} given by a hull function
\[
u_n=h(n\cdot\omega)
\]
where $\omega\in \mathbb{R}$ and
\[h(\theta)=\theta+\tilde{h}(\theta)\] with
\[\tilde{h}(\theta)=\sum_{k\in\mathbb{Z}^d}\hat{h}_k e^{2\pi i k\cdot \theta\alpha}\] equivalently,
\[
\tilde{h}(\theta)=\hat{h}(\theta\alpha)
\]where $\hat{h}:\mathbb{T}^d\rightarrow\mathbb{R}$ is a function
\[
\hat{h}(\sigma)=\sum_{k\in\mathbb{Z}^d}\hat{h}_k e^{2\pi i k\cdot
\sigma}.
\]
We denote the set of $\tilde{h}$ of this type by $QP(\alpha)$. Later
on, we always use the notation $\sigma=\theta\alpha$ for variables
in $\mathbb{T}^d$.

Then \eqref{equilibrium} is equivalent to
\begin{equation}\label{equilibriumhull}
h(\theta+\omega)+h(\theta-\omega)-2h(\theta)+\partial_\alpha
V(\alpha\cdot h(\theta))=0
\end{equation} where $\partial_\alpha V\equiv (\alpha\cdot
\nabla)V$. We write \eqref{equilibriumhull} in terms of $\hat{h}$
which is
\begin{equation}\label{equilibriumhullforhat}
\hat{h}(\sigma+\omega\alpha)+\hat{h}(\sigma-\omega\alpha)-2\hat{h}(\sigma)+\partial_\alpha
V(\sigma+\alpha\cdot \hat{h}(\sigma))=0
\end{equation}

\subsection{External forces}
We will find it convenient to add an external parameter to the
equilibrium equation \eqref{equilibriumhullforhat} and to generalize
to situations where the forces do not derive from a potential. So we
will be looking for solutions of
\begin{equation}\label{hullforce}
\E[\hat{h},\lambda](\theta)\equiv
\hat{h}(\sigma+\omega\alpha)+\hat{h}(\sigma-\omega\alpha)-2\hat{h}(\sigma)+\widehat{U}(\sigma+\alpha\cdot
\hat{h}(\sigma))+\lambda=0
\end{equation} where $\widehat{U}:\mathbb{T}^d\rightarrow
\mathbb{R}$ and $\hat{h}$ is as before, $\lambda\in\mathbb{R}$. Note
that, both $\hat{h}$ and $\lambda$ are unknown.

Adding the extra term $\lambda$ will allow us to study the equation
\eqref{hullforce} for arbitrary periodic $\widehat{U}$ even if
$\widehat{U}$ is not the gradient of some periodic function. Later,
a very simple argument developed in Section \ref{vanishing lemma}
Lemma \ref{lemma:vanishing}, will show that, when $\widehat{U}$ has
a variational structure, then $\lambda=0$.
Hence, in the case that $\widehat{U} = \partial_\alpha \widehat{V}$,
the equation \eqref{hullforce} is equivalent to
\eqref{equilibrium}.

\begin{remark}
This procedure of adding an extra parameter and showing that it
vanishes when there are geometric properties similar to the extra
parameter method introduced in \cite{Moser'67} to prove ``translated
curve theorems'' and developed later in \cite{Yoccoz92,Sevryuk}. The
main advantage is that the extra term allows a more efficient
iterative procedure. The proofs of existence of perturbative
expansions are also considerably easier.
\end{remark}

Our main results will be ``a-posteriori" theorems (see Theorem
\ref{main} and ~\ref{main2} for more details) which show that if we
are given a pair $[\hat{h},~\lambda]$ that solves \eqref{hullforce}
very approximately (and provided $\omega,~\alpha$ satisfy some
appropriate Diophantine condition and $[\hat{h},~\lambda]$ satisfy
some non-degeneracy condition), then there is a true solution close
by. We note that we do not need that the system is ``close to
integrable'', we only need that there is an approximate solution of
the equilibrium equations. Of course, in the case that the system is
close to integrable, the solutions in the integrable case are
approximate solutions. One can, however obtain approximate solutions
in other cases, for example, the result can validate numerically
produced approximate solutions.

The method of proof is based on an iterative procedure of Nash-Moser
type using the quadratic convergence to overcome the small divisors.
The quadratic convergence will be based on some geometric
cancellations that come from the variational structure of the
problem.

We note that the method of proof leads to very  efficient
algorithms, which we present in Section \ref{algorithm} and which
are implemented in \cite{numeric}. We note that the method is based
on the Lagrangian proof in \cite{Rana,LM'01} and in \cite[Section
5.3]{Llave'01}. In \cite{Rafael'08}, it is shown that, in the
classical case, when $V$ is periodic, the method of proof extends to
interactions which are not nearest neighbors. We hope that such an
extension is possible in the quasi-periodic case considered here.
This extension to interactions which are not nearest neighbors is
interesting because the most physical models involve such
interactions.

The precise formulation of Theorem \ref{main} and Theorem
\ref{main2} requires us to make precise the sizes of functions,
Diophantine condition, etc. which we take up in Section
\ref{spaces}.

Note that the solution of \eqref{hullforce} are not unique. If
$[\hat{h}(\sigma),~\lambda]$ is a solution, for any
$\beta\in\mathbb{R}$, $[\hat{h}(\sigma+\beta\alpha)+\beta,~\lambda]$
is a solution. Hence, by choosing $\beta$, we can always choose our
solution normalized in such a way that
\begin{equation}\label{normalization}
\lim_{T\rightarrow
\infty}\frac{1}{T}\int_{-T}^{T}\tilde{h}(\theta)d\theta\equiv
\int_{\mathbb{T}^d}\hat{h}(\sigma)d\sigma=0.
\end{equation}
Note that the choice of $\beta$ that accomplishes
\eqref{normalization} is unique.

We will indeed establish that the solution of \eqref{hullforce} and
\eqref{normalization} is unique.

\section{Function spaces and preliminary estimates}
In Section \ref{spaces}, we collect several standard definitions of
spaces and present some preliminary results on these spaces. In
Section \ref{Diophantine properties} we present definitions of the
Diophantine properties we will use in this paper. In Section
\ref{cohomology} we present well known estimates for cohomology
equations, which are the basis of the KAM procedure.

\subsection{Spaces}\label{spaces}
Given a function
\[
\tilde{h}(\theta)=\sum_{k\in\mathbb{Z}^d}\hat{h}_k e^{2\pi i k\cdot
\alpha\theta}
\] it is natural to consider the function
\[
\hat{h}(\sigma)=\sum_{k\in\mathbb{Z}^d}\hat{h}_k e^{2\pi i k\cdot
\sigma}
\]where $\sigma\in \mathbb{C}^d/\mathbb{Z}^d$.

Clearly, if the above sums converge pointwise
(in our applications they will converge in much
stronger senses), we have
\begin{equation}\label{relation}
\tilde{h}(\theta)=\hat{h}(\alpha\theta)
\end{equation}
We will find it more convenient to define spaces and norms using the
function $\hat{h}$.

Following \cite{CallejaL} we will find it convenient to use at the
same time spaces of analytic functions  and Sobolev spaces. Using
the abstract results in \cite{CallejaL} this leads automatically to
several interesting corollaries such as bootstrap of regularity of
solutions and a numerically verifiable criterion for the breakdown
of analytic solutions. See Section \ref{further}.

\begin{definition}
If the series $\hat{h}(\eta)=\sum_{k\in \mathbb{Z}^d} \hat{h}_k
e^{2\pi i k\cdot \eta}$ defines an analytic function on
$D_\rho\equiv\{~\eta~|~|Im(\eta_i)|<\rho\}$ which extends
continuously to $\overline{D_\rho}$, we denote
\begin{equation}\label{rhonorm}
\|\hat{h}\|_\rho=\sup_{\eta\in\overline{D_\rho}}|\hat{h}(\eta)|
\end{equation}
We denote by $\A_\rho$ the space of such functions. As it
is well known, $\A_\rho$ is a Banach space when endowed
with \eqref{rhonorm}.
Clearly for $\rho' < \rho$,  $\A_\rho \subset \A_{\rho'}$.
By the maximum principle,
 $\| f \|_\rho$ is monotone increasing in
$\rho$ for $f\in\A_\rho$.

For the sake of convenience we will also introduce $\A^r_\rho$, the
Banach space of functions whose $r$ derivative is in $\A_\rho$,
endowed with the norm $\| f\|_{\A^r_\rho} = \| D^r f \|_{\A_\rho} +
| \langle f\rangle | $ where $\langle f\rangle$ denotes the average
of $f$. By Cauchy bounds, if $\rho' < \rho$, then $\A^r_{\rho'}
\subset \A_\rho^r$.

We also define
\begin{equation}\label{Hrnorm}
\|\hat{h}\|_{H^r}^2=\sum_{k\in\mathbb{Z}^d}|\hat{h}_k|^2(1+|k|^2)^r
\end{equation} and denote by $H^r$ the spaces of functions for
which \eqref{Hrnorm} is finite. It is well known that $H^r$ is a
Banach space and indeed a Hilbert space.
\end{definition}



\subsubsection{Some elementary properties}
\label{sec:elementary}

There are several well-known properties of those spaces. We just
note
\begin{itemize}
\item Cauchy estimates for analytic functions:
\begin{align}
&\|D^l\hat{h}\|_{\rho-\delta} \leq C(l,d)\cdot
\delta^{-l}\cdot\|\hat{h}\|_\rho\nonumber\\
& |\hat{h}_k| \leq e^{-2\pi
\cdot|k|\cdot\rho}\cdot\|\hat{h}\|_\rho\nonumber.
\end{align}

\item Interpolation inequalities:
\begin{itemize}
\item Analytic case:
\begin{lemma}[Hadamard 3-circle theorem, see \cite{Stein}]\label{interpolation inequality for
analytic}
\[
\|\hat{h}\|_\rho \leq
\|\hat{h}\|_{\rho-\delta}^{\frac{1}{2}} \cdot
\|\hat{h}\|_{\rho+\delta}^{\frac{1}{2}}.
\]
\end{lemma}
\item Sobolev case:
\begin{lemma}[See \cite{Zeh75,Stein} ]\label{interpolation inequality for
Sobolev}
For any $0\leq n\leq j,~0\leq \theta\leq 1$, denote
$l=(1-\theta)n+\theta j$, we have for any $\hat{h}\in H^j$:
\begin{equation*}
\|\hat{h}\|_l\leq C_{n,j}\cdot
\|\hat{h}\|_n^{1-\theta}\cdot\|\hat{h}\|_j^\theta.
\end{equation*}
\end{lemma}
\end{itemize}

\item Banach algebra properties:
\begin{itemize}
\item Analytic case:
\[
\forall ~\hat{g},\hat{h}\in\mathscr{A}_\rho:~\|\hat{g}\cdot
\hat{h}\|_\rho \leq \|\hat{g}\|_\rho\cdot\|\hat{h}\|_\rho.
\]

\item Sobolev case (see \cite{Adams}): Let $m>\frac{d}{2}$, there exists a constant K
depending only on $m,~d$ such that for any $u,~v\in H^m$, $u\cdot
v\in H^m$ and we have
\[
\|u\cdot v\|_{H^m}\leq K\cdot \|u\|_{H^m}\cdot \|v\|_{H^m}.
\]
\end{itemize}

\item Composition properties:
\begin{itemize}
\item Analytic case (see \cite{LlaveO99,Rafael'08}): Let $\hat{h}\in\mathscr{A}_\rho$ and $\Omega\subseteq\mathbb{C}^d$ be a compact set. Take
$\iota=dist(\mathbb{C}^d-\Omega,(Id+\alpha\cdot\hat{h})(\overline{D_\rho}))$.
\begin{lemma}\label{composition lemma}
Let $\widehat{U}:\Omega\rightarrow \mathbb{C}$ be an analytic
function $\|\widehat{U}(z)\|_{L^\infty(\Omega)} \leq M$. Define
the operator $\Phi$ acting on analytic functions
by
$( \Phi[ \hat{h}])(\sigma)=\widehat{U}(\sigma+\alpha\cdot
\hat{h}(\sigma))$.

We present sufficient conditions that ensure that the
operator is well defined and differentiable.
\begin{itemize}
\item If $\|\hat{h}^*-\hat{h}\|_\rho \leq \iota$, then
$\Phi_{\hat{h}^*}\in \mathscr{A}_\rho$.
\item If $\|\hat{h}^*-\hat{h}\|_\rho<\frac{\iota}{2}$, then
\[
(D\Phi[\hat{h}]\widehat{\Delta})(\sigma)=\partial_\alpha
\widehat{U}(\sigma+\alpha\hat{h}(\sigma))\widehat{\Delta}(\sigma),
\]
\[
\|\Phi[ \hat{h}^*]-\Phi[\hat{h}] -
D\Phi[\hat{h}](\hat{h}^*-\hat{h})\|_\rho
\leq C\cdot \|\hat{h}^*-\hat{h}\|^2_\rho.
\]
\end{itemize}
\end{lemma}

\item Sobolev case (see \cite{Taylor,CallejaL}):

The following result is a consequence of the Gagliardo-Nirenberg
inequalities.

\begin{lemma}\label{lem:composition}
Let $f\in C^m$ and assume $f(0)=0$. Then, for $u\in H^m\cap
L^\infty$
\[
\|f(u)\|_{H^m}\leq K_2(\|u\|_{L^\infty})(1+\|u\|_{H^m}),
\]where $K_2(\lambda)=\sup_{|x|\leq \lambda,~\mu\leq m}|D^\mu
f(x)|$.

In the case that $m>\frac{d}{2}$, if $f\in C^{m+2}$, we have that
\begin{equation}\label{Taylor-sobolev}
\|f\circ(u+v)-f\circ u-Df\circ u\cdot v\|_{H^m}\leq
C_{d,m}(\|u\|_{L^\infty})(1+\|u\|_{H^m})\|f\|_{C^{m+2}}\|v\|_{H^m}^2.
\end{equation}
\end{lemma}
\end{itemize}

The reason for \eqref{Taylor-sobolev} is that we have
pointwise
\[
\begin{split}
f\circ(u + v)(x)& - f\circ u(x) - Df\circ u(x)\cdot v(x) =\\
&=  \int_0^1
dt \, \int_0^t ds \,  D^2 f\circ( u + ts v)(x)\cdot u(x)\cdot v(x)
\end{split}
\]
We obtain the desired result using that, by Gagliardo-Nirenberg
$D^2f \circ (u + ts v) \in H^m$ and that the $H^m$ norm is a Banach
algebra under multiplication.
Therefore, we can estimate the $H^m$ norm of the integrand
independently of $s,t$.

Note that Lemma~\ref{lem:composition} also gives a formula for the
derivative of $\Phi$ as the composition with the derivative of $\U$.
It follows that if $\U \in C^{m+2+\ell}$, then, $\Phi \in C^{\ell
+1}$.

\end{itemize}

\begin{remark}
In \cite{Zeh75} it is shown that the interpolation inequalities
are a consequence of the existence of smoothing operators.
which play an important role in the abstract formulation of
KAM theorem.

We also note that Lemma~\ref{composition lemma} is rather
elementary. It suffices to show that there are Taylor estimates with
uniform constants. (As shown in \cite{LlaveO99}, it is important
that in the domain of $\U$, among any two points $x_1, x_1$, it is
possible to choose a path $\gamma$ such that $\ell(\gamma)$, its
length satisfies $\ell(\gamma) \le C | x_1 - x_2|$.)

Actually, one can prove something stronger, namely that the operator
$\Phi$ is an analytic operator  when $\U$ is analytic, but we will
not need as much.
\end{remark}

The following lemma is standard in the theory of quasi-periodic
functions.
\begin{lemma}\label{averages}
Let $\alpha$ be irrational. Assume that
$\sum_{k\in\mathbb{Z}^d}|\hat{h}_k|<\infty$. (So that, by
Weierstrass M-test, the series
$\tilde{h}(\theta)=\sum_{k\in\mathbb{Z}^d}\hat{h}_k e^{2\pi i k\cdot
\alpha\theta}$ converges uniformly over the real line.) Then,
\[
\hat{h}_0=\lim_{T\rightarrow\infty} \frac{1}{2T}\int_{-T}^T
\tilde{h}(\theta)d\theta.
\]
\end{lemma}
\begin{proof}
Given $\epsilon>0$, we can find $N$ such that
\[\sum_{|k|>N}|\hat{h}_k|\leq \frac{\epsilon}{3}.\]
Hence, for all $T>0$
\[
\left|\frac{1}{2T}\int_{-T}^T\sum_{|k|>N}\hat{h}_ke^{2\pi i k\cdot
\alpha \theta}d\theta\right|\leq\frac{\epsilon}{3}.
\]
Furthermore,
\begin{equation}\label{final}
\frac{1}{2T} \int_{-T}^T \sum_{|k|\leq N}\hat{h}_k e^{2\pi i
k\cdot\alpha \theta}d\theta =\hat{h}_0+\sum_{0<|k|\leq N}\hat{h}_k
\frac{e^{2\pi i k\cdot\alpha T}-e^{-2\pi i k\cdot\alpha T}}{2\pi i
k\cdot \alpha\cdot 2T}
\end{equation}
We see that, for all $T$ sufficiently large, we can assume that the
term in the sum in \eqref{final} (it is a finite sum) is smaller
than $\frac{\epsilon}{3}$. This ends the proof of Lemma
\ref{averages}.
\end{proof}

\subsection{Diophantine properties}\label{Diophantine properties}
Given $\alpha\in\mathbb{R}^d$ such that
\begin{equation}\label{conditon for alpha}
|\alpha\cdot k| \geq \mu |k|^{-\upsilon}, \qquad \forall k\in
\mathbb{Z}^d-\{0\}
\end{equation} where $|k|=|k_1|+|k_2|+...+|k_d|$,
 we are interested in the numbers
$\omega\in\mathbb{R}$ such that
\begin{equation}\label{Diophantine}
|\omega\alpha\cdot k-n|\geq \nu|k|^{-\tau}, \qquad \forall k\in
\mathbb{Z}^d-\{0\},~n\in\mathbb{Z}.
\end{equation}
That is, we are interested in the $\omega\in\mathbb{R}$ such that
$\omega\alpha$ is a Diophantine vector in the standard sense. Here
$\mu,~\upsilon,~\nu,~\tau$ are positive numbers.

We will denote the set of $\omega$ satisfying \eqref{Diophantine} as
$\mathscr{D}(\nu,\tau;\alpha)$. We also denote
$\mathscr{D}(\tau;\alpha)=\cup_{\nu>0}\mathscr{D}(\nu,\tau;\alpha)$.

The abundance of Diophantine numbers in subspaces is a subject of
great current interest in number theory
\cite{Kleinbockconference,Kleinbock}. Nevertheless in our case, it
suffices the following elementary result Lemma~\ref{density} (which
is weaker with respect to the exponents obtained than the results of
\cite{Kleinbockconference,Kleinbock}).

\begin{lemma}\label{density}
If $\alpha\in \mathbb{R}^d$ satisfies \eqref{conditon for alpha} and
$\tau>d+\upsilon$, then $\mathscr{D}(\tau;\alpha)$ is of full
Lebesgue measure.
\end{lemma}
\begin{proof}
Let $A>0$. Consider the sets
\[
\mathscr{B}_{k,n}=\{\omega|~|\omega\alpha\cdot k-n|\leq
\nu|k|^{-\tau}\}=\{\omega|~|\omega-\frac{n}{\alpha\cdot k}|\leq
\frac{\nu|k|^{-\tau}}{|\alpha\cdot k|}\}
\] and
\[
\mathscr{B}_{k,n,A}=\mathscr{B}_{k,n}\cap[A,1.01A].
\]
Clearly,
\[
[A,1.01A]\setminus\mathscr{D}(\nu,\tau;\alpha)=\cup_{k,n}\mathscr{B}_{k,n,A}
\]
So Lemma \ref{density} will be proved when we show
\[
|\cup_{k,n}\mathscr{B}_{k,n,A}|\leq\nu\cdot
C(A,\tau,\alpha,\upsilon,\mu).
\]
We clearly have that $\mathscr{B}_{k,n}$ is an interval of length
\[
|\mathscr{B}_{k,n}|=\frac{2\nu|k|^{-\tau}}{|k\cdot\alpha|}.
\]
We also observe that $\mathscr{B}_{k,n,A}=\varnothing$ unless
$A+\frac{\nu|k|^{-\tau}}{|k\cdot\alpha|}\leq
\frac{n}{k\cdot\alpha}\leq
1.01A+\frac{\nu|k|^{-\tau}}{|k\cdot\alpha|}$ or
$A-\frac{\nu|k|^{-\tau}}{|k\cdot\alpha|}\leq
\frac{n}{k\cdot\alpha}\leq
1.01A-\frac{\nu|k|^{-\tau}}{|k\cdot\alpha|}$. Also clearly,
\[
\sharp\{n|~\mathscr{B}_{k,n,A}\neq \varnothing\}\leq 0.02A\cdot
|k\cdot\alpha|+2
\]
Hence
\begin{align}
|\cup_{k,n}\mathscr{B}_{k,n,A}|\leq& \sum_{\substack{k,n\\
\mathscr{B}_{k,n,A}\neq\varnothing}}\frac{2\nu|k|^{-\tau}}{|k\cdot\alpha|}
\leq  \sum_k \frac{2\nu |k|^{-\tau}}{|k\cdot \alpha|}\cdot
(0.02A\cdot |k\cdot \alpha|+2)\nonumber\\
\leq & 0.02A\cdot 2\nu\cdot\sum_k |k|^{-\tau}+\frac{4\nu}{\mu}\sum_k
|k|^{-(\tau-\upsilon)} \leq \nu\cdot C.\nonumber
\end{align}

\end{proof}

\subsection{Cohomology equations}\label{cohomology}
To prove our results, as it is standard in KAM theory,
we have to study equations of the form
\begin{equation}\label{cohomology equation}
\tilde{\phi}(\theta+\omega)-\tilde{\phi}(\theta)=\tilde{\eta}(\theta)
\end{equation}
where $\tilde{\eta}\in QP(\alpha),~\omega\in\mathbb{R}$ are given
and $\tilde{\phi}$ is the unknown.

We note that if we use the function
$\hat{\phi}(\alpha\theta)=\tilde{\phi}(\theta)$ we see that
\eqref{cohomology equation} is equivalent to
\begin{equation}\label{cohomology equation extended}
\hat{\phi}(\sigma+\omega\alpha)-\hat{\phi}(\sigma)=\hat{\eta}(\sigma).
\end{equation}

The operator solving these equations is unbounded, but it satisfies
some ``tame" estimates from one space to another that can be
overcome by a quadratically convergent algorithm.

Clearly, a necessary solution for the existence of solutions of
\eqref{cohomology equation extended} is
\[
\int \hat{\eta}(\sigma)d\sigma=0.
\]

The equation \eqref{cohomology equation extended} has been
considered in KAM theory. The optimal results for analytic functions
were proved in \cite{Russmann'75} and we just reproduce the results
adapted to our notation. The results for Sobolev regularity are very
easy.

\begin{lemma}\label{lemma:cohomology}
Let $\hat{\eta}\in\mathscr{A}_\rho$ (resp. $H^r,~r\geq \tau$) be
such that
\[
\int_{\mathbb{T}^d}\hat{\eta}(\sigma)d\sigma=0.
\]
Then, there exists a unique solution $\hat{\phi}$ of
\eqref{cohomology equation extended} which satisfies
\begin{equation}\label{normalization equation}
\int_{\mathbb{T}^d}\hat{\phi}(\sigma)d\sigma=0.
\end{equation}
This solution satisfies for any $\rho^\prime<\rho$
\begin{equation}\label{cohomology bounds for analytic}
\|\hat{\phi}\|_{\rho^\prime}\leq C(d,\tau)\cdot
\nu^{-1}\cdot(\rho-\rho^\prime)^{-\tau}\|\hat{\eta}\|_\rho
\end{equation}
(resp.
\begin{equation}\label{cohomology bounds for Sobolev}
\|\hat{\phi}\|_{H^{s-\tau}}\leq C\cdot
\nu^{-1}\cdot\|\hat{\eta}\|_{H^s},\qquad \tau\leq s\leq r).
\end{equation}
Furthermore, any distribution solution of \eqref{cohomology
equation} differs from the solution claimed before by a constant.
\end{lemma}
The method of proof is standard. We take the Fourier coefficients
and see that \eqref{cohomology equation} is equivalent to:
\[
\hat{\phi}_k(e^{2\pi i k\cdot \alpha\omega}-1)=\hat{\eta}_k.
\]
When $k \ne 0$, we can use that $|e^{2\pi i k\cdot
\alpha\omega}-1|^{-1} \le C\cdot {\rm dist}(k \cdot \alpha \omega,
\mathbb{Z})^{-1}$. {From} this, the result for Sobolev space follows
rather straightforwardly. The result for analytic functions --
 with the optimal exponent quoted above -- requires somewhat more
elaborate arguments that use not only the upper bounds in the
denominators, but also that they are not saturated very often. We
refer to \cite{Russmann'75} for the original proof and to
\cite{Llave'01} for a more pedagogical exposition.

\section{Statement of the main results}
\subsection{Statement of the analytic result}
\begin{theorem}\label{main}
Let $h=Id+\tilde{h}$ where $\tilde{h}(\theta)=\hat{h}(\alpha\cdot
\theta)=\sum_{k\in\mathbb{Z}^d}\hat{h}_k \cdot e^{2\pi i\cdot k\cdot
\alpha\theta}$ with $\hat{h}_0=0,~\hat{h}\in\A^1_\rho$
and $\alpha\in\mathbb{R}^d$ such that
$\alpha\cdot j\neq0,~j\in\mathbb{Z}^d-\{0\}$. Denote
$\hat{l}=1+\partial_\alpha \hat{h}$ and
$T_{-\omega\alpha}(\sigma)=\sigma-\omega\alpha$. We assume
\begin{itemize}
\item [(H1)] Diophantine properties \eqref{Diophantine}: $|\omega\alpha\cdot k-n|\geq \nu|k|^{-\tau}, \quad \forall k\in
\mathbb{Z}^d-\{0\},~n\in\mathbb{Z}$.

\item [(H2)] Non-degeneracy condition: $\|\hat{l}(\sigma)\|_\rho\leq N^+, ~\|(\hat{l}(\sigma))^{-1}\|_\rho\leq
N^-$ and $|\langle\frac{1}{\hat{l}\cdot\hat{l}\circ
T_{-\omega\alpha}}\rangle|\geq c$ for some positive constant $c$ where
$\langle f\rangle$ denotes the average of the  periodic function $f$.

\item [(H3)] $\|\E[\hat{h},\lambda]\|_\rho\leq \epsilon$ where $\E$ is defined in \eqref{hullforce}.

\item [(H4)] Composition condition: Take
$\iota=dist(\mathbb{C}^d-\Omega,(Id+\alpha\cdot\hat{h})(\overline{D_\rho}))$
where $\Omega$ is the domain of $\widehat{U}$.
We assume that
$\|\hat{h}\|_\rho+\rho\leq \frac{1}{2}\iota$.
\end{itemize}
Assume furthermore that $\epsilon\leq
\epsilon^*(N^-,N^+,d,\tau,c,\iota,\|\widehat{U}\|_{C^2(\Omega)})\cdot\nu^4\cdot
\rho^{4\tau+A}$ where $\epsilon^*>0$ is a function and $A\in
\mathbb{R}^+$ ( we will make explicit $\epsilon^*$ and $A$ along the proof).

Then,
there exists a periodic function $\hat{h}^*$ and
$\lambda^*\in\mathbb{R}$ such that
\[\E[\hat{h}^*,\lambda^*]=0.\]
Moreover,
\begin{align}
&\|\hat{h}-\hat{h}^*\|_{\frac{\rho}{2}} \leq C\cdot \nu^{-2}
\rho^{-2\tau-A}\cdot\|\E[\hat{h},\lambda]\|_\rho,\nonumber\\
&|\lambda-\lambda^*| \leq
C\cdot\|\E[\hat{h},\lambda]\|_\rho.\nonumber
\end{align}
The solution $[\hat{h}^*,\lambda^*]$ is the only solution of
$\E[\hat{h}^*,\lambda^*]=0$ with zero average for $\hat{h}^*$ in a
ball centered at $\hat{h}$ in $\mathscr{A}_{\frac{3\rho}{8}}$, i.e.
$[\hat{h}^*,\lambda^*]$ is the unique solution in the set
\[
\left\{\hat{g}\in\mathscr{A}_{\frac{3\rho}{8}}~|~\langle\hat{g}\rangle=0,~\|\hat{g}-\hat{h}\|_{\frac{3\rho}{8}}\leq
\frac{\nu^2\cdot\rho^{2\tau}}{2\tilde{C}(N^-,N^+,d,\tau,c,C)}\right\}
\]where $\tilde{C}$ will be made explicit along the proof.

\end{theorem}

\subsection{Statement of the Sobolev result}
\begin{theorem}\label{main2}
Let $m>\frac{1}{2}+2\tau$ and $\widehat{U}\in C^{m+34\tau+1}$. Let
$h=Id+\tilde{h}$ where $\tilde{h}(\theta)=\hat{h}(\alpha\cdot
\theta)=\sum_{k\in\mathbb{Z}^d}\hat{h}_k \cdot e^{2\pi i\cdot k\cdot
\alpha\theta}$ with $\hat{h}\in H^{m+32\tau},~\hat{h}_0=0$ for any
$\theta\in \mathbb{R}$ and $\alpha\in\mathbb{R}^d$ such that
$\alpha\cdot j\neq0,~j\in\mathbb{Z}^d-\{0\}$. Denote
$\hat{l}=1+\partial_\alpha \hat{h}$ and
$T_{-\omega\alpha}(\sigma)=\sigma-\omega\alpha$. We assume
\begin{itemize}
\item [(H1)] Diophantine properties \eqref{Diophantine}: $|\omega\alpha\cdot k-n|\geq \nu|k|^{-\tau}, \quad \forall k\in
\mathbb{Z}^d-\{0\},~n\in\mathbb{Z}$.

\item [(H2)] Non-degeneracy condition: $\|\hat{l}(\sigma)\|_{H^m}\leq N^+, ~\|(\hat{l}(\sigma))^{-1}\|_{H^m}\leq
N^-$ and $|\langle\frac{1}{\hat{l}\cdot\hat{l}\circ
T_{-\omega\alpha}}\rangle|\geq c$ for some positive constant c.

\item [(H3)] $\|\E[\hat{h},\lambda]\|_{H^{m}}\leq \epsilon$.

\end{itemize}
Assume furthermore that $\epsilon\leq
\epsilon^*(N^-,N^+,d,\tau,c,\|\widehat{U}\|_{C^{m+34\tau+1}})\cdot\nu^{-2}$
where $\epsilon^*>0$ is a function which we will make explicit along
the proof. Then, there exists a periodic function $\hat{h}^*\in
H^{m-4\tau}$ and $\lambda^*\in\mathbb{R}$ such that
\[\E[\hat{h}^*,\lambda^*]=0.\]
Moreover,
\begin{align}
&\|\hat{h}-\hat{h}^*\|_{H^{m-4\tau}} \leq C\cdot \nu^{-2}
(N^+)^2\cdot\epsilon,\nonumber\\
&|\lambda-\lambda^*| \leq C\cdot \epsilon.\nonumber
\end{align}
The solution $[\hat{h}^*,\lambda^*]$ is the only solution of
$\E[\hat{h}^*,\lambda]=0$ with zero average for $\hat{h}^*$ in a
neighborhood of $[\hat{h},\lambda]$ in $H^{m+4\tau}$, i.e.
$[\hat{h}^*,\lambda^*]$ is the unique solution in the set
\[
\left\{\hat{g}\in
H^{m+4\tau}~|~\langle\hat{g}\rangle=0,~\|\hat{g}-\hat{h}\|_{H^{m+4\tau}}\leq
\frac{\nu^2}{2\tilde{C}(N^-,N^+,d,\tau,c,\|\widehat{U}\|_{C^{m+34\tau+1}})\cdot
C_{m-4\tau,m+4\tau}}\right\}
\]where $\tilde{C}$ will be made explicit along the proof.

\end{theorem}

\begin{remark} The theorems \ref{main} and \ref{main2} have
the  \emph{``a-posteriori''} format of numerical analysis.

Given a function which solves very approximately the invariance
equations, then there is a true solution nearby.
The needed approximation is quantified in terms of
the non-degeneracy condition  $N^+, N^-, c$,
(which in numerical
analysis are often called \emph{condition numbers} ).
We note that the condition numbers  can be computed in the
approximate solution. We emphasize that this formulation does
not require that the system is close to integrable.

Of course, the non-degeneracy conditions depend on the function $\h$
(and on the parameter of the domain $\rho$ in the analytic case. If
this can cause confusion -- e.g. when we are performing an iterative
step -- we will use $N^\pm(\h; \rho)$.
\end{remark}

\section{Proof of main theorems, Theorem \ref{main} and Theorem \ref{main2}}
As indicated in the introduction, the proofs of Theorems \ref{main}
and \ref{main2} are  based on an iterative step that given an
approximate solution of \eqref{invariance approximate} will produce
a better approximation.

The crux of the proof is to show that, if started with a
sufficiently approximate solution, the procedure converges.

As it is very well known in KAM theory, there are arguments that
establish the convergence, provided that we show that the iterative
procedure satisfies \emph{tame quadratic estimates}. That is, that
the norm of the  new error is bounded by the square of another norm
of the original error (in a smoother space) times a factor that
depends on the ``loss of regularity''.  There are several abstract
theorems in this direction, one which is quite well adapted to the
problem at hand and which we will use appears in \cite{CallejaL}.

In Section \ref{motivation}, we will give some motivations for the
iterative procedure. (It is a Newton method with a small
modification that does not affect the quadratic convergence.)

In Section \ref{algorithm}, we will formulate the iterative
procedure as a succession of elementary sub-steps. We note that
these elementary sub-steps can be implemented by very efficient
algorithms. If the functions $\hat{h}$ are discretized using $N$
points and appropriate algorithms are used for the mathematical
sub-steps, the iterative procedure requires only $O(N)$ storage and
$O(N\log{N})$ operations.

In Section \ref{estimates}, we present estimates for the iterative
step. We first present estimates on how much it changes the
function. Then, we use the Taylor estimates to show that the error
of the improved function  is \emph{tame quadratic}  in the sense of
Nash-Moser theory.

In Section~\ref{convergence}, we will review the convergence of the
procedure, which, as we have indicated before is rather standard,
indeed in \cite{CallejaL}, there is an abstract theorem designed to
cover exactly the problems considered here. Nevertheless, for the
sake of completeness in the analytic case, we will also present a
very short direct proof of the convergence argument. Since this
direct proof is so direct it also leads to  good numerical values.

In Section~\ref{uniqueness}, we present several considerations that
allow us to discuss uniqueness.

Finally, in Section~\ref{further}, we show that we can obtain
several consequences combining the Sobolev and analytic versions of
the a-posteriori theorem. Namely, we show that for analytic mappings
all sufficiently smooth solutions are analytic, that Lindstedt
series converge, and we present a numerically efficient criterion
for the breakdown of analyticity. As it was shown in \cite{ALD}, the
breakdown of analyticity is associated to the onset of transport
properties, so there is some interest in its computation.

Of course, readers interested only in rigorous proofs can safely
skip Section \ref{motivation} and the algorithmic considerations in
Section \ref{algorithm}. Readers only interested in algorithms can
skip \ref{estimates}.

\subsection{Motivation for the iterative step}\label{motivation}
We start from an approximate solution of \eqref{hullforce},
\begin{equation}\label{invariance approximate}
\hat{h}(\sigma+\omega\alpha)+\hat{h}(\sigma-\omega\alpha)-2\hat{h}(\sigma)+\widehat{U}(\sigma+\alpha\cdot
\hat{h}(\sigma))+\lambda=e
\end{equation}
where $e$ is to be thought of as \emph{``small''}.

Our goal is to devise a procedure that gives a much more approximate
solution. For the moment, we will not make precise the sense in which
quantities are small. This will be taken up in Section~\ref{estimates}.


Given an approximate solution as in \eqref{invariance approximate},
the Newton method would consist in finding a solution of
\begin{equation}\label{Newton}
\widehat{\Delta}(\sigma+\omega\alpha)+\widehat{\Delta}(\sigma-\omega\alpha)-2\widehat{\Delta}(\sigma)+
\partial_\alpha\widehat{U}(\sigma+\alpha\cdot
\hat{h}(\sigma))\cdot\widehat{\Delta}(\sigma)+\delta=-e.
\end{equation}
Then $[\hat{h}+\widehat{\Delta},~\lambda+\delta]$ will be a better
approximate solution.

The equation \eqref{Newton} is hard to study because of the term
$\partial_\alpha\widehat{U}(\sigma+\alpha\cdot
\hat{h}(\sigma))\cdot\widehat{\Delta}(\sigma)$ which is not constant
coefficients.

The key observation is that, if we are given
\eqref{invariance approximate} we are also given the
following equation which is just
obtained by taking the derivative of \eqref{invariance approximate}
with respect to $\theta$ (we recall that  $\sigma=\theta\alpha$):
\begin{equation}\label{derivative invariance approximate}
\partial_\alpha\hat{h}(\sigma+\omega\alpha)+\partial_\alpha\hat{h}(\sigma-\omega\alpha)-2\partial_\alpha\hat{h}(\sigma)+
\partial_{\alpha}\widehat{U}(\sigma+\alpha\cdot \hat{h}(\sigma))\cdot
(1+\partial_\alpha \hat{h}(\sigma))=e^\prime(\theta).
\end{equation}

Denoting $\hat{l}(\sigma)=1+\partial_\alpha\hat{h}(\sigma)$, we
rewrite \eqref{derivative invariance approximate}:
\begin{equation}\label{derivative invariance approximate for l}
\hat{l}(\sigma+\omega\alpha)+\hat{l}(\sigma-\omega\alpha)-2\hat{l}(\sigma)+
\partial_{\alpha}\widehat{U}(\sigma+\alpha\cdot \hat{h}(\sigma))\cdot
\hat{l}(\sigma)=e^\prime(\theta).
\end{equation}

If we substitute \eqref{derivative invariance approximate for l}
into \eqref{Newton}, we obtain that the Newton procedure
\eqref{Newton} is equivalent to
\begin{equation}\label{Newton2}
\widehat{\Delta}(\sigma+\omega\alpha)+\widehat{\Delta}(\sigma-\omega\alpha)+\frac{e^\prime(\theta)-
\hat{l}(\sigma+\omega\alpha)-\hat{l}(\sigma-\omega\alpha)}{\hat{l}(\sigma)}\cdot
\widehat{\Delta}(\sigma)=-e-\delta.
\end{equation}

The equation \eqref{Newton2} is still hard to solve, but, as we will
now see, we can study the equation obtained by omitting the term
$e^\prime(\theta)\cdot\widehat{\Delta}(\sigma)$ in \eqref{Newton2}.

Because both $e^\prime(\theta)$ and $\widehat{\Delta}(\sigma)$ are
small, we can hope and we will show that omitting the term
$e^\prime\cdot\widehat{\Delta}(\sigma)$ does not affect the
quadratic character of the procedure.

The key observation is that the quasi-Newton equation for
$\Delta,~\delta$:
\begin{equation}
\label{QuasiNewtonall}
\hat{l}(\sigma)\cdot\widehat{\Delta}(\sigma+\omega\alpha)+\hat{l}(\sigma)\cdot\widehat{\Delta}(\sigma-\omega\alpha)
-\widehat{\Delta}(\sigma)[\hat{l}(\sigma+\omega\alpha)+\hat{l}(\sigma-\omega\alpha)]=(-e-\delta)\cdot
\hat{l}(\sigma)
\end{equation}
 is
equivalent to the system
\begin{equation}\label{QuasiNewton1}
\left( \frac{\widehat{\Delta}}{\hat{l}} \right) \circ
T_{-\omega\alpha}-
\left( \frac{\widehat{\Delta}}{\hat{l}} \right)
= \frac{\widehat{W}}{\hat{l}\cdot\hat{l}\circ T_{-\omega\alpha}}
\end{equation}
\begin{equation}\label{QuasiNewton2}
\widehat{W}\circ T_{\omega\alpha}-\widehat{W}=\hat{l}\cdot
(e+\delta).
\end{equation}

In fact, from \eqref{QuasiNewton1}, we get
\[
\widehat{W}=\widehat{\Delta}\circ T_{-\omega\alpha}\cdot
\hat{l}-\widehat{\Delta}\cdot \hat{l}\circ T_{-\omega\alpha}
\] and
\[
\widehat{W}\circ T_{\omega\alpha}=\widehat{\Delta}\cdot \hat{l}\circ
T_{\omega\alpha}-\widehat{\Delta}\circ T_{\omega\alpha}\cdot \hat{l}
\]
By \eqref{QuasiNewton2}, we can easily get the equivalence.

\medskip

The key point is that the equations \eqref{QuasiNewton1} and
\eqref{QuasiNewton2} are of the form  \eqref{cohomology equation},
and can be studied
using the theory in Section~\ref{cohomology}.

We write $\widehat{W}=\widehat{W}^0+\overline{\widehat{W}}$ where
$\widehat{W}^0$ is a function with  zero average and
$\overline{\widehat{W}}$ is a number. That is, we decompose
$\widehat{W}$ into its average and the zero average part. Both are
unknowns. Now we describe the procedure to solve the system
\eqref{QuasiNewton1}, \eqref{QuasiNewton2}.

We first choose $\delta$ to be the unique value that makes  the
average of the right-hand-side of \eqref{QuasiNewton2} zero. Then,
we can apply Lemma~\ref{lemma:cohomology} to find $\widehat{W}^0$
solving \eqref{QuasiNewton2}. We note that there is only one choice
of $\delta$ and then, $\widehat{W}^0$ is determined uniquely, by the
condition that it solves \eqref{QuasiNewton2} and that it has zero
average. The only solutions of \eqref{QuasiNewton2} differ from it
by a constant.

Then  we observe that
$\overline{\widehat{W}}=-\frac{\langle\frac{\widehat{W}^0}{\hat{l}\cdot
\hat{l}\circ
T_{-\omega\alpha}}\rangle}{\langle\frac{1}{\hat{l}\cdot\hat{l}\circ
T_{-\omega\alpha}}\rangle}$ is the only possible value of the
average of solutions of \eqref{QuasiNewton2} that  makes the
right-hand-side of \eqref{QuasiNewton1} with zero average. Them, we
can apply again Lemma~\ref{lemma:cohomology} to find
$\frac{\widehat{\Delta}}{\hat{l}}$ solving \eqref{QuasiNewton1}.
This solution is unique up to the addition of constant. Once we have
$\frac{\widehat{\Delta}}{\hat{l}}$, we obtain $\widehat{\Delta}$ is
obtained just multiplying by $\hat{l}$. Note that the $\Delta$ is
thus determined uniquely up to the addition of a constant multiple
of $\hat{l}$. In particular, $\widehat{\Delta}$ is unique when we
impose the normalization that it has zero average.

\subsection{Formulation of the iterative step}\label{algorithm}

\begin{algorithm}
\label{alg:quasi-newton}
Given $\hat{h}:\mathbb{T}^d\rightarrow \mathbb{R},~\lambda\in
\mathbb{R}$ with $\hat{h}(\sigma)=\sum_{k\in \mathbb{Z}^d}\hat{h}_k
e^{2\pi i k\sigma}$ and $\tilde{h}(\theta)=\hat{h}(\alpha\theta)$
for $\theta\in\mathbb{R}$ and any irrational vector
$\alpha\in\mathbb{R}^d,~d\geq2$, we will calculate:
\begin{itemize}
\item [1)] Let $\mathscr{L}=\hat{h}(\sigma+\alpha\omega)+\hat{h}(\sigma-\alpha\omega)-2
\hat{h}(\sigma)$. In Fourier components
$\widehat{\mathscr{L}}_k=2(\cos{\omega \alpha\cdot k}-1)\hat{h}_k$.

\item [2)] We calculate $\widehat{U}(\sigma+\alpha\cdot\hat{h}(\sigma))$.

\item [3)] So we can calculate $e=\mathscr{L}+\widehat{U}+\lambda$.

\item [4)] Calculate $\hat{l}=1+\partial_\alpha\hat{h}$. In Fourier components $\hat{l}_k=\delta_{k,0}+2\pi i k\cdot
\alpha\cdot \hat{h}_k$  where $\delta_{0,k}$ is the Kronecker delta.

\item [5)] Let $f=\hat{l}\cdot e$.

\item [6)] Choose $\delta=-\langle f\rangle$.

\item [7)] Denote $b=\hat{l} \cdot (e+\delta)$.

\item [8)] Solve the cohomology equation \eqref{QuasiNewton2} for
$\widehat{W}^0$ with zero average. That is,
$\widehat{W}^0_k=\frac{b_k}{2(\cos{\omega\alpha\cdot k}-1)}$.

\item [9)] Take $\overline{\widehat{W}}=-\frac{\langle\frac{\widehat{W}^0}{\hat{l}\cdot \hat{l}\circ
T_{-\omega\alpha}}\rangle}{\langle\frac{1}{\hat{l}\cdot\hat{l}\circ
T_{-\omega\alpha}}\rangle}.$

\item [10)] Calculate $\widehat{W}=\widehat{W}^0+\overline{\widehat{W}}$.

\item [11)] Solve the cohomology equation \eqref{QuasiNewton1}.
Find $\tilde \beta$  satisfying
$\tilde \beta - \tilde \beta \circ T_{-\omega \alpha} =
\frac{\widehat{W}}{\hat{l}\cdot\hat{l}\circ T_{-\omega\alpha}}$.
That is,
$\tilde \beta_k=\frac{a_k}{2(\cos{\omega\alpha\cdot k}-1)}$ where
$a=\frac{\widehat{W}}{\hat{l}\cdot \hat{l}\circ T_{-\omega\alpha}}$.

\item[12)] We obtain  $\Delta = (\tilde \beta  + \bar \beta)\cdot \hat{l}$
where $\bar \beta$ is chosen to be
$ - \langle \tilde \beta \cdot \hat{l}\rangle/ \langle \hat l \rangle$
so that $\langle \Delta \rangle  = 0$.

\end{itemize}
\end{algorithm}

\begin{remark}
\label{rem:uniqueness} It  is important to note that the procedure
also shows that the solution of \eqref{QuasiNewton1} and
\eqref{QuasiNewton2} is unique up to the addition of a constant
multiple of $\hat{l}$ to $\hat{\Delta}$.

Hence, if we the choose the solutions of
\eqref{QuasiNewton1}, \eqref{QuasiNewton2}
which satisfy the normalization $\langle \widehat{\Delta} \rangle = 0$,
the solutions are unique.

Of course, since the system \eqref{QuasiNewton1}, \eqref{QuasiNewton2}
is equivalent to \eqref{QuasiNewtonall}, the same considerations apply
to \eqref{QuasiNewtonall}.
\end{remark}

\begin{remark}
\label{efficiency}

Note that if we consider functions discretized by their values at
$N$ points and by $N$ Fourier coefficients, steps 1), 4), 8), 11)
are fast in Fourier coefficients (they require $O(N)$ operations)
while the other steps are fast (they require $O(N)$ operations) in
the representation of the function by its values at points. Of
course, once we know the representation in space or in Fourier
coefficients we can use the Fast Fourier transform which requires
$O(N\log N)$ operations to compute the other.

Note also that if we discretize the function as above, the iterative
step only requires to store several functions, and therefore we only
need to store $O(N)$ numbers.

We, thus obtain a quadratical convergent algorithm, with $O(N)$
storage requirements and $O(N\log(N))$ operations. In contrast, a
Newton method would require $O(N^2)$ storage to store a matrix and
$O(N^3)$ operations to solve the linear equations (there are faster
algorithms \cite{Knuth} to solve linear equations but they do not
seem to be practical). In practice the present algorithm with $N =
10^7$ can run comfortably on a modest desktop machine.
\end{remark}

\begin{remark}
\label{rem:approximate-inverse} It is important to note that
$[\widehat{\Delta}, \delta]$, the outcome of the algorithm depends
linearly on $e \equiv \E[ \hat h, \lambda]$.

Hence, we will write
\begin{equation} \label{eq:approximate-inverse}
[\widehat \Delta, \delta] = \eta[\hat h, \lambda] e
\end{equation}

The operator $\eta$ is called an \emph{``approximate right inverse''}
in Nash-Moser theory. See, for example \cite{Zeh75}.

Notice that the estimates for the improved solution can be written
in a symbolic way as   estimating $\E[ [\hat{h}, \lambda] +
\eta[\hat{h}, \lambda] \E[\hat{h}, \lambda] \, ]$, which, using
Taylor expansion (up to quadratic errors) becomes
\begin{equation} \label{linearapprox}
\E[\hat h, \lambda] +
 D\E[ \hat{h}, \lambda]  \eta[\hat{h}, \lambda] \E[\hat{h}, \lambda]
\end{equation}
In the Newton method, we would choose $\eta$ in such a way that
\eqref{linearapprox} vanishes.
As pointed out in \cite{Moser66a, Zeh75}, it suffices
that the norm of  \eqref{linearapprox}  can be bounded
by the square of another norm of $\E[\hat h, \lambda]$.
\end{remark}

\subsection{Estimates on the quasi-Newton step}\label{estimates}
In this section we show that the Quasi-Newton method specified in
Algorithm~\ref{alg:quasi-newton} produces more approximate
solutions. We will present two versions of the estimates, one in
analytic spaces and another one in Sobolev spaces. The goal is  to
obtain that the new error is quadratic in the original error even if
in a weaker norm. We note that the analytic estimates presented are
a bit more delicate and involve a  condition,
\eqref{inductionassumption}.

\subsubsection{Some  useful identities}
We start by remarking an elementary identity that will be used for
both  the  analytic and the  Sobolev estimates:
\begin{equation}\label{geometric identity}
\hat{l}\cdot(D_1
\E[\hat{h},\lambda]\widehat{\Delta})-\widehat{\Delta}\cdot(D_1
\E[\hat{h},\lambda]\hat{l})=-\hat{l}\cdot(\E[\hat{h},\lambda]+\delta).
\end{equation}where $D_1$ denote the derivative with respect to the
first variable.

We also have the following identity obtained
just adding and subtracting terms in the definition of
$\E[ \hat h + \hat \Delta, \lambda + \delta] $ and
grouping them.
\begin{equation}
\label{eq:errorestimate}
\begin{split}
&\E[\hat{h}+\widehat{\Delta},\lambda+\delta] \\
=& \E[\hat{h},\lambda]+\widehat{\Delta}(\sigma+\omega\alpha)+
\widehat{\Delta}(\sigma-\omega\alpha)-2\widehat{\Delta}(\sigma)+
\delta\\
 &+\widehat{U}(\sigma+\alpha\cdot(\hat{h}+\widehat{\Delta})(\sigma))-\widehat{U}(\sigma+\alpha\cdot\hat{h}(\sigma))\\
=&\E[\hat{h},\lambda]+(-\E[\hat{h},\lambda])+
\frac{\hat{l}(\sigma+\omega\alpha)+
\hat{l}(\sigma-\omega\alpha)-2\hat{l}(\sigma)}{\hat{l}(\sigma)}\cdot
\widehat{\Delta}(\sigma)\\
 &+\widehat{U}(\sigma+\alpha\cdot(\hat{h}+\widehat{\Delta})(\sigma))-\widehat{U}(\sigma+\alpha\cdot\hat{h}(\sigma))\\
=&\frac{e^\prime(\theta)-\partial_\alpha\widehat{U}(\sigma+\alpha\cdot\hat{h}(\sigma))\cdot\hat{l}(\sigma)}{\hat{l}(\sigma)}\cdot\widehat{\Delta}(\sigma)
+\widehat{U}(\sigma+\alpha\cdot(\hat{h}+\widehat{\Delta})(\sigma))-\widehat{U}(\sigma+\alpha\cdot\hat{h}(\sigma))\\
=&e^\prime\cdot\frac{\widehat{\Delta}(\sigma)}{\hat{l}(\sigma)}+\widehat{U}(\sigma+\alpha\cdot(\hat{h}+\widehat{\Delta})(\sigma))-
 \widehat{U}(\sigma+\alpha\cdot\hat{h}(\sigma))-\partial_\alpha
 \widehat{U}(\sigma+\alpha\cdot\hat{h})\cdot
\widehat{\Delta}(\sigma)\\
\equiv&e^\prime\cdot
\frac{\widehat{\Delta}(\sigma)}{\hat{l}(\sigma)}+R
\end{split}
\end{equation}
where
$R=\widehat{U}(\sigma+\alpha\cdot(\hat{h}+\widehat{\Delta})(\sigma))-
 \widehat{U}(\sigma+\alpha\cdot\hat{h}(\sigma))-\partial_\alpha
 \widehat{U}(\sigma+\alpha\cdot\hat{h})\cdot\widehat{\Delta}(\sigma)$.
Clearly, $R$ is the remainder of the Taylor estimate in the composition
studied in Lemma~\ref{composition lemma} and Lemma~\ref{lem:composition}.

\subsubsection{Estimates for the iterative step in analytic spaces}
We now observe that for any $\rho^\prime<\rho$, by \eqref{cohomology
bounds for analytic}, we obtain using \eqref{QuasiNewton2}
\[
\|\widehat{W}^0\|_{\rho^\prime}\leq C(d,\tau)\cdot
\nu^{-1}\cdot(\rho-\rho^\prime)^{-\tau}\cdot N^+\cdot\|e\|_\rho.
\]

Since the average of $\widehat W$ is obtained in
step 9) of Algorithm \eqref{alg:quasi-newton}, we
have the estimate for $\overline{\widehat{W}}$:
\[
|\overline{\widehat{W}}|\leq
c\cdot\|\widehat{W}^0\|_{\rho^\prime}\cdot (N^-)^2\leq
c\cdot(N^-)^2\cdot N^+\cdot
C(d,\tau)\cdot\nu^{-1}\cdot|\rho-\rho^\prime|^{-\tau}\cdot
\|e\|_\rho.
\]

Therefore, we obtain the estimates for $\widehat{W}$:
\[
\|\widehat{W}\|_{\rho^\prime} \leq M \cdot \nu^{-1} \cdot
(\rho-\rho^\prime)^{-\tau}\cdot\|e\|_\rho
\] where $M=(c\cdot(N^-)^2+1)\cdot N^+ C(d,\tau)$.
The important point is that the
constant is uniform provided $\hat{h}$ stays in a neighborhood in
$\|\cdot\|_\rho$ norm.

Again for $\rho^{\prime\prime}<\rho^\prime$, by \eqref{cohomology
bounds for analytic}, we have
\[
\|\tilde{\Delta}\|_{\rho^{\prime\prime}}\leq C(d,\tau)\cdot
\nu^{-1}\cdot(\rho^\prime-\rho^{\prime\prime})^{-\tau}\cdot
(N^-)^2\cdot\|\widehat{W}\|_{\rho^\prime}\leq M^\prime\cdot
(\rho-\rho^\prime)^{-\tau}\cdot(\rho^\prime-\rho^{\prime\prime})^{-\tau}\cdot\|e\|_\rho.
\]
So, we have
\[
\|\tilde \beta \|_{\rho^{\prime\prime}} \leq
N^+\cdot M \cdot
\nu^{-2} (\rho-\rho^\prime)^{-\tau}
\cdot(\rho^\prime-\rho^{\prime\prime})^{-\tau}\cdot\|e\|_\rho\]
and
\begin{equation}
\label{Deltaestimates} \|\Delta \|_{\rho^{\prime\prime}} \leq
N^+\cdot  N^{-} M \cdot \nu^{-2} (\rho-\rho^\prime)^{-\tau}
\cdot(\rho^\prime-\rho^{\prime\prime})^{-\tau}\cdot\|e\|_\rho.
\end{equation}

Similarly, using Cauchy estimates, we
obtain for $\rho''' < \rho''$
\begin{equation}
\label{Deltaprimeestimates} \|\Delta \|_{\rho^{\prime\prime}} \leq
N^+\cdot  N^{-} M \cdot \nu^{-2} (\rho'' - \rho'')^{-1}
(\rho-\rho^\prime)^{-\tau}
\cdot(\rho^\prime-\rho^{\prime\prime})^{-\tau}\cdot\|e\|_\rho.
\end{equation}

If we take $\rho -\rho' = \rho' - \rho''$
in \eqref{Deltaestimates}
and $\rho  - \rho'' = \rho' -\rho'' = \rho -\rho''$
(and redefine $\rho'''$ ) in \eqref{Deltaprimeestimates}
 we obtain:
\begin{equation} \label{stepestimates}
\begin{split}
 \|\widehat{\Delta}\|_{\rho^{\prime\prime}} \leq
M^\prime\cdot
\nu^{-2} (\rho-\rho'')^{-2\tau}
\cdot\|e\|_\rho\\
\|\widehat{\Delta}\|_{\rho^{\prime\prime}} \leq M^\prime\cdot
\nu^{-2} (\rho-\rho'')^{-2\tau-1} \cdot\|e\|_\rho.\\
\end{split}
\end{equation}

If we have that $ \| \widehat{\Delta}\|_{\rho''}  \le \iota/2$,
which, by \eqref{stepestimates}  is implied by
\begin{equation}
\label{inductionassumption}
M^\prime\cdot
\nu^{-2} (\rho-\rho'')^{-2\tau}
\cdot\|e\|_\rho
\leq \iota/2,
\end{equation}
 we can define $\U(\sigma + \alpha \widehat{h} +
\widehat{\Delta}(\sigma) )$ and indeed apply Taylor's estimate
we obtain
\begin{equation} \label{firstpart}
\|R\|_{\rho^{\prime\prime}} \leq \sup{\widehat{U}}\cdot
\|\widehat{\Delta}\|_{\rho^{\prime\prime}}^2 \le M \nu^{-4} (\rho''
- \rho)^{-2 \tau} \| e \|_\rho^2.
\end{equation}

The first term in the right-hand-side of \eqref{eq:errorestimate} is
estimated using the Cauchy estimates and the previous estimates on
$\| \widehat{\Delta}\|$.
\[
\begin{split}
\big| \big|
e^\prime\cdot\frac{\widehat{\Delta}}{\hat{l}}
\big| \big|_{\rho''}
& \le (\rho -\rho'')^{-1} \| e\|_{\rho} N^-
N^+\cdot M^\prime\cdot
\nu^{-2} (\rho-\rho^{\prime \prime})^{-2\tau}
\cdot\|e\|_\rho \\
& = M (\rho - \rho'')^{-2 \tau -1} \nu^{-2} \| e\|_\rho^2 \\
& \le M (\rho - \rho'')^{-4 \tau -1} \nu^{-4} \| e\|_\rho^2.
\end{split}
\]
The last estimate is  done with the purpose of simplifying the expressions,
but it is obviously wasteful. Note that $\tau \ge 1$ and that
the estimates above  are delicate only when  $\rho -\rho''$, $\nu$
are small.

Finally, putting together the estimates for the two terms in the
right-hand-side of \eqref{eq:errorestimate}, we have:
\begin{equation}\label{eq:goodestimates}
\|\E[\hat{h}+\widehat{\Delta},\lambda+\delta]\|_{\rho^{\prime\prime}}\leq
C\cdot
\nu^{-4}(\rho-\rho^{\prime\prime})^{-4\tau}\|\E[\hat{h},\lambda]\|_\rho^2.
\end{equation}

Therefore, we have proved the following
\begin{lemma} \label{iterative}
In the hypothesis of Theorem~\ref{main}.

Assume that \eqref{inductionassumption}  holds.
Then, the improved function obtained applying
Algorithm \eqref{alg:quasi-newton},
satisfies
\eqref{eq:goodestimates}.
\end{lemma}

As it is well-known in KAM theory,  the
above estimates imply that the iterative procedure can be repeated
indefinitely and the resulting sequence converges to a function
satisfying the claims of Theorem \ref{main}.
Indeed, the paper \cite{CallejaL} contains an abstract theorem
that immediately applies to this situation.
We will discuss this in more detail in Section~\ref{convergence}.

\subsubsection{Sobolev estimates for the iterative step.}
Let $s>\frac{d}{2}$. According to the algorithm 8), we have
\begin{align}
\|b\|_{H^s}&=\sum_{k\in\mathbb{Z}^d}(1+|k|^2)^s|\hat{l}_k\cdot(e_k+\delta_{k,0})|^2\nonumber\\
&=\sum_{k\in\mathbb{Z}^d-\{0\}}(1+|k|^2)^s\cdot|\hat{l}_k\cdot
e_k|^2\nonumber\\
&\leq \|f\|_{H^s}\leq K\cdot
\|\hat{l}\|_{H^s}\cdot\|e\|_{H^s}.\nonumber
\end{align}
 By \eqref{cohomology bounds for
Sobolev}, we obtain
\[
\|\widehat{W}^0\|_{H^{s-\tau}}\leq C\cdot
\nu^{-1}\cdot\|b\|_{H^s}\leq C\cdot \nu^{-1}\cdot K\cdot N^+
\cdot\|e\|_{H^s}.
\]
We get
\[
|\overline{\widehat{W}}| \leq c\cdot (N^-)^2
\|\widehat{W}^0\|_{H^{s-\tau}}
\] and
\[
\|\widehat{W}\|_{H^{s-\tau}} \leq C\cdot \nu^{-1}\cdot K\cdot N^+
\cdot\|e\|_{H^s}.
\]

So we have
\[
\|\tilde{\Delta}\|_{H^{s-2\tau}}\leq C\cdot \nu^{-1} K^2 \cdot
(N^-)^2\cdot \|\widehat{W}\|_{H^{s-\tau}}\leq C\cdot \nu^{-2}\cdot
N^+\cdot (N^-)^2\cdot\|e\|_{H^s}.
\]
Hence,
\[
\|\widehat{\Delta}\|_{H^{s-2\tau}}\leq
C\cdot\nu^{-2}\cdot(N^+)^2\cdot(N^-)^2\cdot\|e\|_{H^s}.
\]
We recall that  the approximate inverse of the derivative
 $\eta[\hat{h},\lambda]$ is just the result of applying
applying the algorithm in Section \ref{algorithm}, i.e.
$[\widehat{\Delta}, \delta] =
\eta[\hat{h},\lambda]\E[\hat{h},\lambda]$. We have proved the
following lemma:
\begin{lemma}
Let $s>\frac{d}{2}+2\tau$. Then we have
\[
\|\eta[\hat{h},\lambda]\E[\hat{h},\lambda]\|_{H^{s-2\tau}}\leq
C\cdot\nu^{-2}\cdot(N^+)^2\cdot(N^-)^2\cdot\|e\|_{H^s}.
\]
\end{lemma}

We will also need estimates on $(D_1
\E[\hat{h},\lambda]\eta[\hat{h},\lambda]-Id)(\E[\hat{h},\lambda]+\delta)$.
\begin{lemma}
\begin{align}
&\|(D_1
\E[\hat{h},\lambda]\eta[\hat{h},\lambda]-Id)(\E[\hat{h},\lambda]+\delta)\|_{H^{s-2\tau}}\nonumber\\
\leq& C\cdot \nu^{-2}\cdot(N^+)^2\cdot (N^-)^3\cdot
\|\E[\hat{h},\lambda]\|_{H^{s-2\tau-1}}\|\E[\hat{h},\lambda]\|_{H^s}.\nonumber
\end{align}
\end{lemma}
\begin{proof}
By the definition of $\eta[\hat{h},\lambda]$, we know
$\widehat{\Delta}=-\eta[\hat{h},\lambda](\E[\hat{h},\lambda]+\delta)$.
Hence,
\begin{align}
&(D_1
\E[\hat{h},\lambda]\eta[\hat{h},\lambda]-Id)(\E[\hat{h},\lambda]+\delta)\nonumber\\
=&-D_1\E[\hat{h},\lambda]\widehat{\Delta}-\E[\hat{h},\lambda]-\delta
=\frac{\widehat{\Delta} \cdot D_1
\E[\hat{h},\lambda]}{\hat{l}}.\nonumber
\end{align}
So we have that
\begin{align} &\|(D_1
\E[\hat{h},\lambda]\eta[\hat{h},\lambda]-Id)(\E[\hat{h},\lambda]+\delta)\|_{H^{s-2\tau}}\nonumber\\
\leq& C\cdot \nu^{-2}\cdot(N^+)^2\cdot
(N^-)^3\cdot\|\E[\hat{h},\lambda]\|_{H^{s-2\tau-1}}\|\E[\hat{h},\lambda]\|_{H^s}.\nonumber
\end{align}
\end{proof}

\subsection{Convergence of the procedure}
\label{convergence}

The existence of solutions both in the analytic case and in the
Sobolev case is deduced from the estimates in
Section~\ref{estimates} followed by Nash-Moser estimates.

Indeed in \cite{CallejaL}, one can find an
 abstract Nash-Moser implicit function
theorem which is tailored to the theorems \ref{main} and
\ref{main2}.

In this section, we reproduce the theorem from \cite{CallejaL} and
explain why it is applicable.  We note that the the theorem has
several corollaries which are of physical interest and we present
them in Section~\ref{further}.

For the sake of completeness,
in Section~\ref{convergence-analytic}, we present
a direct proof of the convergence in the
analytic case.

\subsubsection{An abstract implicit function theorem}

In \cite[Appendix A]{CallejaL} one can find a proof
of the following result, Theorem~\ref{implicitfunctiontheorem}.
This is an abstract theorem that applies to operators in
scales of Banach spaces, which have smoothing operators.

In \cite{CallejaL} one can also find a verification that
the Sobolev spaces  and analytic spaces considered
indeed have smoothing operators (one can take
$S_t \sum_k \hat{h}_k e^{2 \pi i k \sigma}   =
 e^{-t | k| } \hat{h}_k e^{2 \pi i k \sigma}$.
The regularity properties of the operator entering in the assumptions
of Theorem~\ref{implicitfunctiontheorem} follow immediately for
the composition properties presented  in Section~\ref{sec:elementary}
and, specially Lemma~\ref{lem:composition}.

\begin{theorem}\label{implicitfunctiontheorem}
Let $m>2\tau$ and $\mathcal{X}^r$ for $m\leq r\leq m+34\tau$ be a
scale of Banach spaces with smoothing operators as shown in
\cite{CallejaL}. Let $\mathcal {B}_r$ be the unit ball in $\mathcal
{X}^r$, $\tilde{\mathcal{B}}_r=\hat{h}+\mathcal{B}_r$ the unit ball
translated by $\hat{h}\in \mathcal{X}^r$ and
$\mathcal{B}(\mathcal{X}^r,\mathcal{X}^{r-2\tau})$ is the space of
bounded linear operators from $\mathcal{X}^r$ to
$\mathcal{X}^{r-2\tau}$. Consider a map $\mathcal
{F}:\tilde{\mathcal{B}}_r \rightarrow \mathcal{X}^{r-2\tau}$ and
$\eta:\tilde{\mathcal{B}}_r \rightarrow
\mathcal{B}(\mathcal{X}^r,\mathcal{X}^{r-2\tau})$ satisfying the
following:
\begin{itemize}
\item [(i)]
$\mathcal{F}(\tilde{\mathcal{B}}_r\cap\mathcal{X}^r)\subset\mathcal{X}^{r-2\tau}$
for $m\leq r\leq m+34\tau$.

\item [(ii)] $\mathcal{F}|_{\tilde{\mathcal{B}}_r\cap\mathcal{X}^r}:\tilde{\mathcal{B}}_r\cap\mathcal{X}^r \rightarrow
\mathcal{X}^{r-2\tau}$ has two continuous Fr\'echet derivatives,
both bounded by some constant $M$, for $m\leq r\leq m+34\tau$.

\item [(iii)] $\|\eta[\widehat{\Delta}]\mathcal{F}[\widehat{\Delta}]\|_{\mathcal{X}^{r-2\tau}} \leq C\cdot \|\mathcal{F}[\widehat{\Delta}]\|_{\mathcal{X}^r},
~\widehat{\Delta}\in\tilde{\mathcal{B}}_r$, for
$r=m-2\tau,~m+32\tau$.

\item [(iv)]
$\|(D\mathcal{F}[\widehat{\Delta}]\eta[\widehat{\Delta}]-Id)\mathcal{F}[\widehat{\Delta}]\|_{\mathcal{X}^{r-2\tau}}
\leq C\cdot
\|\mathcal{F}[\widehat{\Delta}]\|_{\mathcal{X}^r}^2,~\widehat{\Delta}\in\tilde{\mathcal{B}}_r$,
for $r=m$.

\item [(v)]
$\|\mathcal{F}[\widehat{\Delta}]\|_{\mathcal{X}^{m+32\tau}}\leq
C\cdot (1+\|\widehat{\Delta}\|_{\mathcal{X}^{m+34\tau}}),~
\widehat{\Delta}\in\tilde{\mathcal{B}}_m$.
\end{itemize}
Then if $\|\mathcal{F}[\hat{h}]\|_{\mathcal{X}^{m-2\tau}}$ is
sufficiently small, there exists $\hat{h}^*\in \mathcal{X}^m$ such
that $\mathcal{F}[\hat{h}^*]=0$. Moreover,
$\|\hat{h}-\hat{h}^*\|_{\mathcal{X}^m}<C\cdot
\|\mathcal{F}[\hat{h}]\|_{\mathcal{X}^{m-2\tau}}$.
\end{theorem}

We recall that the method of proof of
Theorem~\ref{implicitfunctiontheorem}, following \cite{Schwartz} is
to modify the quasi-Newton step adding a smoothing step. That is,
one constructs a sequence $[\h_{n+1}, \lambda_{n+1}] = [\h_n,
\lambda_n] + S_{t_n} \eta[ \hat h_n, \lambda_n ] \E[ \h_n,
\lambda_n]$.
 The choices of $t_n$
have to be carefully chosen so that the quadratic convergence
(in some norm) is maintained.  The main difference
between Theorem \ref{implicitfunctiontheorem} and the result in
\cite{Schwartz} is that  Theorem \ref{implicitfunctiontheorem} includes
the fact that $\eta$ is an approximate inverse and not an inverse.

The estimates showing that $\eta$ is indeed an approximate inverse
are the estimates obtained in Section~\ref{estimates}.

\subsubsection{A direct proof of the convergence in the
analytic case}
\label{convergence-analytic}

Since the estimates in the analytic case are so easy, we present a
direct proof.  As we will see, the estimates are rather easy to
verify. The main difficulty is the order of the choices.

We start with an  approximate solution  $[\hat h_0, \lambda_0]$ with
$\hat h \in \A^1_{\rho_0}$.

Since we will have to change the function through
an iterative procedure, we note that the
condition numbers $N^+, N^-, c$ depend on the functions we
are considering, nevertheless, they are uniform in
a $\A^1_\rho$ neighborhood.

We start by choosing a number $\gamma > 0$ such that that in
neighborhood of size $\gamma$ in $\A^1_{\rho_0}$, we have that
$N^\pm \le 2 N^{\pm}(\h_0)$, $c \le 2 c(\h_0)$.

The following algebraic identities will be useful
in estimating the change of non-degeneracy conditions in the iterative
step.
\begin{equation} \label{auxiliary}
\begin{split}
N^+(h; \rho_0 ) &\equiv
\| 1 + \partial_\alpha \hat{h} \|_{\rho_0} \\
&\le
N^+(\h_0; \rho_0)
\|  \partial_\alpha (\hat{h} - \hat{h}_0 \|_{\rho_0}
\\
N^-(\h; \rho_0) & \equiv
\| (1 + \partial_\alpha \hat{h})^{-1} \|_{\rho_0}  \\
& \le N^{-}(\h_0; \rho_0) +
\| \partial_\alpha (\hat{h} - \hat{h}_0 )\|_{\rho_0}
N^-(h_0; \rho_0) N^{-}(\h; \rho) \\
\left| c(\h)  - c(\h_0) \right| & =
\left| \left\langle
\frac{ 1 }{\hat{l} \cdot \hat{l}\circ T_{-\omega \alpha}  }
-\frac{ 1 }{\hat{l}_0 \cdot \hat{l}_0\circ T_{-\omega \alpha}}
\right\rangle \right|\\
& = \left| \left \langle
\frac {\hat{l}_0\, ( \hat{l}_0 - \hat{l})\circ T_{-\omega \alpha}
+ (\hat{l} - \hat{l}_0)}{ \hat{l}_0\hat{l}_0 \circ T_{-\omega \alpha}
\hat{l}\hat{l} \circ T_{-\omega \alpha}}
\right \rangle \right| \\
& \le  \left[ N^{-}(\h; \rho_0) N^{-}(\h_0; \rho_0)\right]^2 (\|
\hat{l}\|_{\rho_0} +\| \hat{l}_0\|_{\rho_0}  ) \| \hat{l} -
\hat{l}_0\|_{\rho_0}.
\end{split}
\end{equation}

Hence we can find a number $\gamma > 0$ depending only on the
non-degeneracy conditions $N^\pm, c$ so that all the functions in a
ball of radius $\gamma$ in $\A^1_{\rho_0}$ centered at $\h_0$, have
non-degeneracy constants not larger than twice the non-degeneracy
assumptions of $\h_0$.

More generally, we have, by the same argument that if $\| \h -
\h_0\|_{\A^1_\rho} \le \gamma$, then, $N^\pm(\h; \rho) \le 2
N^{-\pm}(\h_0, \rho)$.

The key estimates are, as follows to show
that, with some convenient choices of radii, which we
do at the outset, the iterative process can be applied indefinitely
and indeed it converges. We will use \eqref{auxiliary} to
show that the non-degeneracy constants do not deteriorate
much.

We denote by
\begin{equation}
\label{rhochoice} \rho_n = \rho_{n-1}  - \frac{\rho_0}{4} 2^{-n} =
\rho_0( 1 - \frac{1}{4}\sum_{i=0}^{n} 2^{-i})
\end{equation}
and provided that we can apply the iterative step (that is, provided
that \eqref{inductionassumption} applies with the choices of
$\rho_n$ in \eqref{rhochoice}, we define  for $n \ge 1$,
$[\h_n, \lambda_n] = [\h_{n-1}, \lambda_{n-1}] + \eta[
\h_{n-1}, \lambda_{n-1}]\E[\h_{n-1},\lambda_{n-1}]$.

We denote by $M$ the constant in Lemma~\ref{iterative} corresponding to
twice the degeneracy assumptions corresponding to the original function.

If \eqref{inductionassumption} applies $n$ times,
for typographical simplicity, we denote
$\ep_i = \| \E[ \hat{h}_i, \lambda_i \|_{\rho_i}$
we see that
\begin{equation}\label{inductivebounds}
\begin{split}
\ep_n &\le M \nu^{-2} \rho_0^{-4 \tau} 2^{(n-1)4\tau } \ep_{n-1}^2
\le  (M \rho_0^{-4 \tau} \nu^{-2})^{1 + 2}  2^{(n-1)4\tau + 2(n-2) 4 \tau}   \ep_{n-2}^{2^2} \\
& \cdots \\
& \le (M \nu^{-2}\rho_0^{-4 \tau})^{1 + 2+ \cdots+ 2^n }
2^{(n-1)4\tau + 2(n-2) 4 \tau + \cdots 2^{n-1}  4 \tau }   \ep_{0}^{2^n} \\
&\le (M \nu^{-2}\rho_0^{-4 \tau})^{2^{n+1}} 2^{8 \tau^{2^n} }
\ep_{0}^{2^n}.
\end{split}
\end{equation}

We see that if $(M \nu^{-2}\rho_0^{-4 \tau})^2 2^{8 \tau}   \ep_{0}
< 1$, the right-hand-side of \eqref{inductivebounds} decreases
faster than any exponential. Indeed the factor can be made as small
as desired by assuming that $\ep_0$ is small enough.

If we apply $n$-times the inductive step, we see that the distance
from the range of $h_n$ to the complement of the domain of
definition of $\U$ is at least \[
\begin{split}
\iota - \sum_{i = 0}^n \| \Delta_i \|_{\rho_n}
& \ge \iota - \sum_{i = 0}^n \| \Delta_i \|_{\rho_i}
\ge \iota - \sum_{i = 0}^n M' \nu^{-2} \rho_0^{-4\tau} 2^{i 4 \tau} \ep_i \\
& \ge \iota - \sum_{i = 0}^n M' \nu^{-2} \rho_0^{-4\tau} (A
\ep_0)^{2^1}.
\end{split}
\]
Note that if $\ep_0$ is small enough, this is bounded from below by
$\iota \frac{3}{4}$ independent of $n$.

According to Lemma~\ref{iterative}, the only thing we have to verify
is \eqref{inductionassumption}, which with the choices of radii that
we have made amounts to:
\[
M^\prime\cdot \nu^{-2} \rho_0^{-2\tau} 2^{n 4 \tau} \ep_n \leq
\iota/4.
\]
We note that this condition is satisfied independently of $n$ if $n$
is large enough.

Using \eqref{auxiliary}, we have:
\[
\begin{split}
N^{+}(\h_n, \rho_n) &\le N^{+}(h_{n-1} , \rho_n) + \|D
\Delta_n\|_{\rho_n} \le N^{+}(h_{n-1}, \rho_{n-1}) + M \nu^{-2}
\rho_0^{4 \tau +1} 2^{4 \tau(n-1)} \ep_{n-1} \\
& \le  N^+(h_0, \rho_0) +
M \nu^{-2} \rho_0^{4 \tau +1} \sum_{i = 0}^{m-1}  2^{4 \tau(n-1)} (A
\ep_0)^{2^i}
\end{split}
\]
and similarly for $N^-,c$. Therefore, under smallness conditions on
$\ep_0$, we get that the non-degeneracy conditions do not change by
a factor $2$ from the original one, so that the induction hypothesis
are satisfied.

In summary, under just three smallness conditions in $\ep_0$,
which can be assessed just looking at the non-degeneracy
conditions, we conclude that the iterative step can
be carried out infinitely often and that
the assumptions on the non-degeneracy constants
make in the estimates for the step remain valid.

We also note that since
$\rho_n \ge \rho_0/2$, we
have
\[
\begin{split}
\|h_N - \h_0\|_{\rho_0/2}  + |\lambda_n - \lambda_0| &\le
 \sum_{n=1}^N \| \h_n - \h_{n-1} \|_{\rho_0/2}  + |\lambda_n - \lambda_{n-1}|\\&\le
\| \h_n - \h_{n-1} \|_{\rho_n } \le \sum_{n=1}^N  (A
\ep_0)^{2^{2^n}} 2^{(4 \tau +1) n} M \nu^{-4}\rho_0^{-4 \tau}
\end{split}
\]
which establishes the  quantitative claims made for
the result.

\subsection{Uniqueness of the solution}\label{uniqueness}
In this section, we establish the uniqueness
claims for the Theorems \ref{main}, \ref{main2}.
We note that the proof is very elementary and only uses
the theory of linearized solutions as well as the interpolation inequalities
in Section~\ref{sec:elementary}.

\subsubsection{Uniqueness for the analytic case}\label{uniqueness for analytic
case} If $\|\hat{h}^*-\hat{h}^{**}\|_{\frac{\rho}{4}},
|\lambda^*-\lambda^{**}|$ is sufficiently small and
$\E[\hat{h}^*,\lambda^*]=\E[\hat{h}^{**},\lambda^{**}]=0$, by
Taylor's theorem and Lemma \ref{composition lemma}, we have
\begin{equation}\label{taylor}
0=\E[\hat{h}^{**},\lambda^{**}]-\E[\hat{h}^*,\lambda^*]=D_1\E[\hat{h}^*,\lambda^*](\hat{h}^{**}-\hat{h}^*)+(\lambda^{**}-\lambda^*)+R
\end{equation} where $\|R\|_{\frac{\rho}{4}}\leq C\cdot \|\hat{h}^{**-}\hat{h}^*\|_{\frac{\rho}{4}}^2$.

Now, denoting as before $\hat{l}=1+\partial_\alpha \hat{h}^*$ and
recalling that
\[
D_1\E[\hat{h}^*,\lambda^*]\cdot
\hat{l}=\frac{d}{d\theta}\E[\hat{h}^*,\lambda^*]=0
\]
we can write the equation \eqref{taylor} as:
\begin{equation}
\hat{l}\cdot
(D_1\E[\hat{h}^*,\lambda^*](\hat{h}^{**}-\hat{h}^*))-(\hat{h}^{**}-\hat{h}^*)\cdot(D_1\E[\hat{h}^*,\lambda^*]
\hat{l})=-\hat{l}R.
\end{equation}
The proof of uniqueness is based on uniqueness of the solution of
the system \eqref{QuasiNewton1} and \eqref{QuasiNewton2}. By the
estimates in Section \ref{estimates}, we conclude that for any
$0<\rho^{\prime\prime}<\rho^\prime<\frac{\rho}{4}$,
\[
\|\hat{h}^{**}-\hat{h}^*\|_{\rho^{\prime\prime}}\leq
C(N^-,N^+,d,\tau,c)\cdot\nu^{-2}\cdot(\rho-\rho^\prime)^{-\tau}\cdot(\rho^\prime-\rho^{\prime\prime})^{-\tau}\|R\|_\rho
\]
Take $\rho^{\prime\prime}=\frac{\rho}{8}$ and
$\rho^\prime=\frac{3}{16}\rho$.


In the analytic case, we obtain
\begin{align}\label{inequality for uniqueness}
&\|\hat{h}^{**}-\hat{h}^*\|_{\frac{\rho}{8}} \leq \tilde{C}\cdot
\nu^{-2}\cdot \rho^{-2\tau}\cdot
\|\hat{h}^{**}-\hat{h}^*\|_{\frac{\rho}{4}}^2\nonumber\\
\leq &\tilde{C}\cdot \nu^{-2}\cdot \rho^{-2\tau}\cdot
\|\hat{h}^{**}-\hat{h}^*\|_{\frac{\rho}{8}}\cdot\|\hat{h}^{**}-\hat{h}^*\|_{\frac{3}{8}\rho},\nonumber
\end{align}
where $\tilde{C}>0$ is a constant depending on $N^-,N^+,d,\tau,c,C$.
The last inequality holds by Lemma \ref{interpolation inequality for
analytic}. So when $\|\hat{h}^*-\hat{h}\|_{\frac{3}{8}\rho}$ small
enough, we obtain $\hat{h}^{**}=\hat{h}^*,~\lambda^{**}=\lambda^*$.
This completes the proof of uniqueness of the solution in Theorem
\ref{main} for the analytic case.


\subsubsection{Uniqueness for the Sobolev case}
Instead of applying Hadamard 3-circle theorem for the analytic case,
we use the interpolation inequality for Sobolev case ( Lemma
\ref{interpolation inequality for Sobolev}).

Following the proof in Section \ref{uniqueness for analytic case},
we will have
\begin{align}
&\|\hat{h}^{**}-\hat{h}^*\|_{H^{m-4\tau}}\leq
\tilde{C}\cdot \nu^{-2}\cdot\|\hat{h}^{**}-\hat{h}^*\|_{H^m}^2 \nonumber\\
\leq& \tilde{C}\cdot\nu^{-2}\cdot
C_{m-4\tau,m+4\tau}\cdot\|\hat{h}^{**}-\hat{h}^*\|_{H^{m-4\tau}}\cdot\|\hat{h}^{**}-\hat{h}^*\|_{H^{m+4\tau}}.\nonumber
\end{align}
This completes the proof of uniqueness of the solution in Theorem
\ref{main2} for the Sobolev case.

\section{vanishing lemma}\label{vanishing lemma}
In this section we prove
\begin{lemma}\label{lemma:vanishing}
Consider a solution of \eqref{hullforce} with the stated periodicity
condition. If
\[
\widehat{U}=\partial_\alpha V
\]
then $\lambda=0$.
\end{lemma}
\begin{proof}
The proof is very simple. We multiply \eqref{hullforce} by
$h^\prime(\theta)$ and compute
$\lim_{T\rightarrow\infty}\frac{1}{2T}\int_{-T}^{T}$ of all the
terms. We note that this produces the formula
\begin{equation}\label{lambda formula}
\lambda=-\lim_{T_i\rightarrow\infty}
\frac{1}{2T}\int_{-T}^{T}\widehat{U}(h(\theta))\cdot
h^\prime(\theta)d\theta.
\end{equation}
In fact, we observe that
\[
h(\theta+\omega)+h(\theta-\omega)-2h(\theta)=\tilde{h}(\theta+\omega)+\tilde{h}(\theta-\omega)-2\tilde{h}(\theta)\in
QP(\alpha).
\]
Similarly,
\[
h^\prime(\theta)=1+\tilde{h}^\prime(\theta)\in QP(\alpha).
\]
Hence,
\[
[h(\theta+\omega)+h(\theta-\omega)-2h(\theta)]\cdot
h^\prime(\theta)\in QP(\alpha)
\] and we have that
\[
\begin{split}
&\lim_{T\rightarrow\infty}\frac{1}{2T}\int_{-T}^T
[h(\theta+\omega)+h(\theta-\omega)-2h(\theta)]\cdot
h^\prime(\theta)\nonumber\\
& =\sum_{k\in\mathbb{Z}^d-\{0\}} -\hat{h}_k\cdot 2(\cos(2\pi\omega
k\cdot\alpha)-1)\cdot\hat{h}_{-k}2\pi i (k\cdot\alpha)+\hat{h}_0
(2\cos(2\pi\omega 0\cdot\alpha)-1)\\
&=0.
\end{split}
\]
The first equality is true because of Lemma \ref{averages} and the
fact that the sum is Cauchy formula for the $k=0$ coefficient of the
integrand ( as we will see below). The fact that the sum is $0$ is
clear because it is antisymmetric in $k$.

In fact, note that
\[
\widehat{[\tilde{h}\circ T_\omega+\tilde{h}\circ
T_{-\omega}-2\tilde{h}]}_k=2(\cos{(2\pi \omega
k\cdot\alpha)}-1)\hat{h}_k,
\] in particular, the coefficient vanishes for $k=0$, and
\[
\widehat{[h^\prime]}_k=\delta_{0,k}+2\pi i
k\cdot\alpha\cdot\hat{h}_k.
\]
We have using Cauchy formula for the Fourier series of the product
\begin{align}
&\widehat{[[\tilde{h}\circ T_\omega+\tilde{h}\circ
T_{-\omega}-2\tilde{h}]\cdot h^\prime]}_0\nonumber\\
=&\sum_{k\in\mathbb{Z}^d}2(\cos(2\pi \omega
k\cdot\alpha)-1)\hat{h}_k\cdot[\delta_{0,k}-2\pi i
k\cdot\alpha\cdot\hat{h}_{-k}]\nonumber\\
=&\sum_{k\in\mathbb{Z}^d-\{0\}} -\hat{h}_k\cdot 2(\cos(2\pi\omega
k\cdot\alpha)-1)\cdot\hat{h}_{-k}2\pi i (k\cdot\alpha)+\hat{h}_0
(2\cos(2\pi\omega 0\cdot\alpha)-1).\nonumber
\end{align}

We also observe that
\[
\int_{-T}^{T}\partial_\alpha V(\alpha h(\theta))\cdot
h^\prime(\theta)d\theta=V(\alpha h(T))-V(\alpha h(-T)).
\]
So it is bounded independent of T. When we divide the integral by
$2T$ and take the limit $T\rightarrow \infty$. We obtain $0$. This
ends the proof of Lemma \ref{lemma:vanishing}.
\end{proof}

\section{Several further  consequences of the formalism}
\label{further}

As pointed out in \cite{CallejaL}, once one has an a-posteriori
theorem with local uniqueness in analytic and Sobolev spaces, there
are more or less automatically several consequences which could  be
of interest for applications and which we now make explicit for our
case.

\subsection{Existence of perturbative expansions to all orders
and their convergence}

If we consider models in which the interaction has a small
parameter, i.e. the interaction term is given by $\ep \U$, it is
interesting to know whether one can write formal power series $\hat
h_\ep = \sum_n \ep^n \hat h^n $, $\lambda_\ep = \sum_n \ep^n
\lambda^n $ which solve \eqref{hullforce} in the sense of power
series as well as the normalization condition \eqref{normalization}.
Furthermore it is interesting to show that that the series
converges. These power series for hull functions are very similar to
the Lindstedt series in mechanics.

We will show that, when the frequencies are Diophantine, the
solution to both questions is affirmative. Series exist to all
orders and converge.

\subsubsection{Existence of Lindstedt series  to all orders}

We first argue that one can find the solution to \eqref{hullforce}
in the sense of power series.

If we substitute the power series and match like powers of $\ep$, we
obtain a hierarchy of equations for the coefficients of the
perturbation. At order $\ep^0$ we obtain:
\begin{equation}\label{order0}
\hat h^0 (\sigma + \omega\alpha) + \hat h^0 (\sigma - \omega\alpha)
- 2 \hat h^0 (\sigma) + \lambda^0 = 0
\end{equation}
which implies that $\lambda^0 = 0$, $\hat h^0$ is a constant.
Because of the normalization \eqref{normalization}, we have $\hat
h^0 = 0$.

At order $\ep^1$, we obtain
\begin{equation}\label{order1}
\hat h^1 (\sigma + \omega\alpha) + \hat h^1 (\sigma - \omega\alpha)
- 2 \hat h^1 (\sigma) + \U( \sigma ) +  \lambda^1 = 0.
\end{equation}
This equation is very similar to the equations studied in
Section~\ref{cohomology}. Indeed, in Fourier series, it is
equivalent to
\[
\hat h^1_k 2(\cos(2 \pi \omega k \cdot \alpha  -1)  = \U_k +
\delta_{0,k}\lambda^1.
\]
We see that we can determine $\lambda^1 = - \U_0$. $\hat h^1_0$ is
not determined by \eqref{order1} but the normalization
\eqref{normalization} sets it to $h^1_0 = 0$. All the other Fourier
coefficients can be determined and indeed if $\U_k$ is analytic in
some domain then $\h_k^1$ is analytic in a slightly smaller domain.

In general, at order $n$, we obtain:
\begin{equation}\label{ordern}
\hat h^n (\sigma + \omega\alpha) + \hat h^n (\sigma - \omega\alpha)
- 2 \hat h^n (\sigma) + R_n(\sigma)  +  \lambda^n = 0
\end{equation}
where $R_n$ is a polynomial expression in $\h^1, \h^2, \cdots,
\h^{n-1}$ with coefficients which are derivatives of $\U$. So, we
can assume by induction that $R_n$ is known. The equation
\eqref{ordern} is of the same form as \eqref{order1} and the same
analysis shows that we can get a unique solution for $\h^n$ and,
hence, recover the hypothesis.

\subsubsection{Convergence of the formal power series} The fact
that the formal power series converges is a very easy consequence of
the fact that there are analytic families $\hat{h}_\ep$,
$\lambda_\ep$, which solve the equations.  This is a general fact,
which is a consequence of the formalism and we go over the proof
rather quickly. See \cite{CallejaL, GEdlL}.

We just note that the method we have used works just as well for
complex functions. We note that for all $\ep$ small enough, there is
a solution. (Note that we can take $\hat h = 0$, $\lambda = 0$
 as an approximate solution if
$\ep$ is small). So, it suffices to show that this solution depends
differentiably on the complex parameter $\ep$. We follow the
standard practice in analysis of first obtaining a guess of the
derivative and then proving it indeed satisfies the definition of
derivative.

For a fixed $\ep$ we can guess $\frac {d}{d\ep}\hat h_\ep$, $\frac
{d}{d\ep} \lambda_\ep$ because if they existed, they should satisfy
\begin{equation} \label{derivativeguess}
\begin{split}
&\frac {d}{d\ep} \hat h_\ep(\sigma + \omega \alpha) + \frac
{d}{d\ep} \hat h_\ep(\sigma -  \omega \alpha) - 2 \frac {d}{d\ep}
\hat h_\ep(\sigma)\\ &+ \ep \partial_\alpha \U(\sigma  + \alpha \hat
h_\ep) \frac {d}{d\ep}\hat  h_\ep(\sigma) + \U( \sigma + \alpha \hat
h(\sigma)) + \frac {d}{d\ep} \lambda_\ep = 0.
\end{split}
\end{equation}
The method used in Section~\ref{motivation} shows that the equation
\eqref{derivativeguess} for$\frac {d}{d\ep} \hat h_\ep(\sigma)$ can
be transformed into a constant coefficient equation (note that, by
assumption $h_\ep$ is an exact solution of \eqref{hullforce}).

Now, to prove that this guess indeed is the derivative, we just note
that $|| \E_{\ep + \mu}( \h_\ep + \mu \frac{d}{d\ep} \h_\ep )
||_{\rho - \eta} \le C |\mu|^2$. Then, applying the a-posteriori
format and the local uniqueness, we conclude $|| \h_{\ep + \mu} -
\h_\ep - \mu \frac{d}{d\ep} \h_\ep ||_{(\rho - \eta)/2 } \le C
\mu^2$.

\subsection{Bootstrap of regularity}\label{bootstrap of regularity}
In this section we state the theorem of bootstrap of regularity and
omit the proof. See \cite{CallejaL} for more details.

\begin{theorem}
Let $\h \in H^m $, $ \lambda$  be a solution of \eqref{hullforce}
with $\widehat{U}$ analytic.

Assume that $m$ is large enough (depending only on  the Diophantine
exponent). Then, $\h$ is analytic.
\end{theorem}

The idea of the proof is very simple. We can take a truncation of
the Fourier series as an approximate solution.  Of course, these are
analytic functions. If the decrease of the Fourier series is fast
enough, it is possible to use the analytic theorem and conclude that
there is an analytic solution. By the uniqueness in Sobolev spaces,
this must be the original solution.

In Sobolev regularity this is restated as Theorem ~6.8 in
\cite{CallejaL}. In \cite{SalamonZ, GEdlL}, one can find a similar
argument for $C^r $ classes. The argument for $C^r$ classes in the
later papers is somewhat more involved since it obtains sharper
results by using better approximations than truncating.

\subsection{A practical numerical criterion for the analyticity breakdown}

The above considerations lead to a very practical and reliable way
to compute thresholds of breakdown of analytic solutions.

Observe that by now, we have efficient algorithms to compute the
invariant tori, given approximate solutions. This, of course,
immediately leads to a continuation algorithm. Since  we have an
a-posteriori theorem, we can be sure that, the approximate solutions
produced numerically (which satisfy the invariance equation up to a
few units of round--off error) correspond to true solutions if they
satisfy the non-degeneracy conditions.

Therefore, a practical algorithm to compute the threshold of
breakdown is to implement the continuation method and monitor the
non-degeneracy conditions.

In many occasions it happens that the only condition that fails is
that $|| \hat h||_{H^m}$ becomes very large. In that case, one can
argue that the solutions experience a breakdown because if there
were analytic tori in a neighborhood of parameters, the Sobolev
norms would remain  bounded. Similar methods for the periodic
Frenkel-Kontorova models and models with long range interactions
have been implemented in \cite{CallejaL}.

Some implementations of the method are already in progress
\cite{numeric} and the results will be reported elsewhere.

\section*{Aknowledgements}
The work of X.S. has been supported by CSC grant 2003619040. Both
authors have been supported by NSF DMS-0911389, and TCB NHRP 0223.
We thank T. Blass and S. Hernandez for several discussions.

\bibliographystyle{alpha}
\bibliography{QPreference}

\begin{thebibliography}{vEFRJ99}

\bibitem[Ada75]{Adams}
Robert~A. Adams.
\newblock {\em Sobolev spaces}.
\newblock Academic Press [A subsidiary of Harcourt Brace Jovanovich,
  Publishers], New York-London, 1975.
\newblock Pure and Applied Mathematics, Vol. 65.

\bibitem[ALD83]{ALD}
S.~Aubry and P.~Y. Le~Daeron.
\newblock The discrete {F}renkel-{K}ontorova model and its extensions. {I}.
  {E}xact results for the ground-states.
\newblock {\em Phys. D}, 8(3):381--422, 1983.

\bibitem[AP10]{Aliste-Prieto}
Jos{\'e} Aliste-Prieto.
\newblock Translation numbers for a class of maps on the dynamical systems
  arising from quasicrystals in the real line.
\newblock {\em Ergodic Theory Dynam. Systems}, 30(2):565--594, 2010.

\bibitem[BHdlL11]{numeric}
T.~Blass, S.~Hernandez, and R.~de~la Llave.
\newblock Computation of quasiperiodic equilibria in quasi-periodic media and
  their breakdown.
\newblock 2011.
\newblock In preparation.

\bibitem[Bur87]{Burkov87}
S.~E. Burkov.
\newblock One-dimensional model of the quasicrystalline alloy.
\newblock {\em J. Statist. Phys.}, 47(3-4):409--438, 1987.

\bibitem[Bur88]{Burkov88}
S.~E. Burkov.
\newblock Ground states of two-dimensional quasicrystals.
\newblock {\em J. Statist. Phys.}, 52(1-2):453--461, 1988.

\bibitem[Bur90]{Burkov90}
S.~E. Burkov.
\newblock Ground states of two-dimensional quasicrystals.
\newblock {\em Phys. Rev. B (3)}, 41(15):10413--10436, 1990.

\bibitem[CdlL10]{CallejaL}
Renato Calleja and Rafael de~la Llave.
\newblock A numerically accessible criterion for the breakdown of
  quasi-periodic solutions and its rigorous justification.
\newblock {\em Nonlinearity}, 23(9):2029--2058, 2010.

\bibitem[dlL01]{Llave'01}
Rafael de~la Llave.
\newblock A tutorial on {KAM} theory.
\newblock In {\em Smooth ergodic theory and its applications ({S}eattle, {WA},
  1999)}, volume~69 of {\em Proc. Sympos. Pure Math.}, pages 175--292. Amer.
  Math. Soc., Providence, RI, 2001.

\bibitem[dlL08]{Rafael'08}
Rafael de~la Llave.
\newblock K{AM} theory for equilibrium states in 1-{D} statistical mechanics
  models.
\newblock {\em Ann. Henri Poincar\'e}, 9(5):835--880, 2008.

\bibitem[dlLO99]{LlaveO99}
R.~de~la Llave and R.~Obaya.
\newblock Regularity of the composition operator in spaces of {H}\"older
  functions.
\newblock {\em Discrete Contin. Dynam. Systems}, 5(1):157--184, 1999.

\bibitem[FK39]{FK}
J.~Frenkel and T.~Kontorova.
\newblock On the theory of plastic deformation and twinning.
\newblock {\em Acad. Sci. URSS, J. Physics}, 1:137--149, 1939.

\bibitem[GEdlL08]{GEdlL}
A.~Gonz{\'a}lez-Enr{\'{\i}}quez and R.~de~la Llave.
\newblock Analytic smoothing of geometric maps with applications to {KAM}
  theory.
\newblock {\em J. Differential Equations}, 245(5):1243--1298, 2008.

\bibitem[GGP06]{Gambaudo}
Jean-Marc Gambaudo, Pierre Guiraud, and Samuel Petite.
\newblock Minimal configurations for the {F}renkel-{K}ontorova model on a
  quasicrystal.
\newblock {\em Comm. Math. Phys.}, 265(1):165--188, 2006.

\bibitem[Kle01]{Kleinbockconference}
Dmitry Kleinbock.
\newblock Some applications of homogeneous dynamics to number theory.
\newblock In {\em Smooth ergodic theory and its applications ({S}eattle, {WA},
  1999)}, volume~69 of {\em Proc. Sympos. Pure Math.}, pages 639--660. Amer.
  Math. Soc., Providence, RI, 2001.

\bibitem[Kle08]{Kleinbock}
Dmitry Kleinbock.
\newblock An extension of quantitative nondivergence and applications to
  {D}iophantine exponents.
\newblock {\em Trans. Amer. Math. Soc.}, 360(12):6497--6523, 2008.

\bibitem[Knu81]{Knuth}
Donald~E. Knuth.
\newblock {\em The art of computer programming. {V}ol. 2}.
\newblock Addison-Wesley Publishing Co., Reading, Mass., second edition, 1981.
\newblock Seminumerical algorithms, Addison-Wesley Series in Computer Science
  and Information Processing.

\bibitem[Koc08]{Koch-boundary}
Hans Koch.
\newblock Existence of critical invariant tori.
\newblock {\em Ergodic Theory Dynam. Systems}, 28(6):1879--1894, 2008.

\bibitem[LM01]{LM'01}
M.~Levi and J.~Moser.
\newblock A {L}agrangian proof of the invariant curve theorem for twist
  mappings.
\newblock In {\em Smooth ergodic theory and its applications ({S}eattle, {WA},
  1999)}, volume~69 of {\em Proc. Sympos. Pure Math.}, pages 733--746. Amer.
  Math. Soc., Providence, RI, 2001.

\bibitem[Mos66]{Moser66a}
J{\"u}rgen Moser.
\newblock A rapidly convergent iteration method and non-linear partial
  differential equations. {I}.
\newblock {\em Ann. Scuola Norm. Sup. Pisa (3)}, 20:265--315, 1966.

\bibitem[Mos67]{Moser'67}
J{\"u}rgen Moser.
\newblock Convergent series expansions for quasi-periodic motions.
\newblock {\em Math. Ann.}, 169:136--176, 1967.

\bibitem[Ran87]{Rana}
David Rana.
\newblock {\em P{ROOF} {OF} {ACCURATE} {UPPER} {AND} {LOWER} {BOUNDS} {TO}
  {STABILITY} {DOMAINS} {IN} {SMALL} {DENOMINATOR} {PROBLEMS}}.
\newblock ProQuest LLC, Ann Arbor, MI, 1987.
\newblock Thesis (Ph.D.)--Princeton University.

\bibitem[RJ97]{Radulescu-Jansen}
O.~Radulescu and T.~Janssen.
\newblock Dynamics of lattice vibrations for one-dimensional commensurate and
  incommensurate composites with harmonic interaction.
\newblock {\em J. Phys. A}, 30(12):4199--4214, 1997.

\bibitem[RJ99]{Radulescu-Jansen2}
O.~Radulescu and T.~Janssen.
\newblock Sliding mode and breaking of analyticity in the double-chain model of
  incommensurate composites.
\newblock {\em Physical Review B}, 60(18):737--745, 1999.

\bibitem[R{\"u}s75]{Russmann'75}
Helmut R{\"u}ssmann.
\newblock On optimal estimates for the solutions of linear partial differential
  equations of first order with constant coefficients on the torus.
\newblock In {\em Dynamical systems, theory and applications ({R}encontres,
  {B}attelle {R}es. {I}nst., {S}eattle, {W}ash., 1974)}, pages 598--624.Lecture
  Notes in Phys., Vol. 38. Springer, Berlin, 1975.

\bibitem[Sch60]{Schwartz}
J.~Schwartz.
\newblock On {N}ash's implicit functional theorem.
\newblock {\em Comm. Pure Appl. Math.}, 13:509--530, 1960.

\bibitem[Sev99]{Sevryuk}
M.~B. Sevryuk.
\newblock The lack-of-parameters problem in the {KAM} theory revisited.
\newblock In {\em Hamiltonian systems with three or more degrees of freedom
  ({S}'{A}gar\'o, 1995)}, volume 533 of {\em NATO Adv. Sci. Inst. Ser. C Math.
  Phys. Sci.}, pages 568--572. Kluwer Acad. Publ., Dordrecht, 1999.

\bibitem[Ste70]{Stein}
Elias~M. Stein.
\newblock {\em Singular integrals and differentiability properties of
  functions}.
\newblock Princeton Mathematical Series, No. 30. Princeton University Press,
  Princeton, N.J., 1970.

\bibitem[SZ88]{SalamonZ}
Dietmar Salamon and Eduard Zehnder.
\newblock Flows on vector bundles and hyperbolic sets.
\newblock {\em Trans. Amer. Math. Soc.}, 306(2):623--649, 1988.

\bibitem[Tay97]{Taylor}
Michael~E. Taylor.
\newblock {\em Partial differential equations. {III}}, volume 117 of {\em
  Applied Mathematical Sciences}.
\newblock Springer-Verlag, New York, 1997.
\newblock Nonlinear equations, Corrected reprint of the 1996 original.

\bibitem[vE99]{vanErpthesis}
T.~S. van Erp.
\newblock {F}renkel {K}ontorova model on quasiperiodic substrate potentials.
\newblock 1999.
\newblock Thesis.

\bibitem[vEF02]{vanErp'02}
T.~S. van Erp and A.~Fasolino.
\newblock {A}ubry transition studied by direct evaluation of the modulation
  functions of infinite incommensurate systems.
\newblock {\em Europhys. Lett.}, 59(3):330--336, 2002.

\bibitem[vEFJ01]{VanErp'01}
T.~S. van Erp, A.~Fasolino, and T.~Janssen.
\newblock Structural transitions and phonon localization in {F}renkel
  {K}ontorova models with quasi-periodic potentials.
\newblock {\em Ferroelectrics}, 250:421--424, 2001.

\bibitem[vEFRJ99]{VanErp'99a}
T.~S. van Erp, A.~Fasolino, O.~Radulescu, and T.~Janssen.
\newblock Pinning and phonon localization in {F}renkel-{K}ontorova models on
  quasiperiodic substrates.
\newblock {\em Physical Review B}, 60(9):6522--6528, 1999.

\bibitem[Yoc92]{Yoccoz92}
Jean-Christophe Yoccoz.
\newblock Travaux de {H}erman sur les tores invariants.
\newblock {\em Ast\'erisque}, (206):Exp.\ No.\ 754, 4, 311--344, 1992.
\newblock S\'eminaire Bourbaki, Vol.\ 1991/92.

\bibitem[Zeh75]{Zeh75}
E.~Zehnder.
\newblock Generalized implicit function theorems with applications to some
  small divisor problems. {I}.
\newblock {\em Comm. Pure Appl. Math.}, 28:91--140, 1975.

\end{thebibliography}
\end{document}